\documentclass[a4paper,UKenglish]{lipics-v2016}
\usepackage{ifthen}
\newcommand{\techRep}{true} 
\newcommand{\iftechrep}{\ifthenelse{\equal{\techRep}{true}}}
\usepackage{microtype}
\usepackage{amssymb}
\usepackage{amsmath}
\usepackage{kbordermatrix}
\usepackage{tikz}
\usetikzlibrary{arrows, automata, positioning, spy, 3d, calc}
\theoremstyle{theorem}
\newtheorem{proposition}[theorem]{Proposition}
\newtheorem{question}[theorem]{Question}
\newtheorem{claim}[theorem]{Claim}
 
\newcommand{\N}{\mathbb{N}}

\newcommand{\Q}{\mathbb{Q}}
\newcommand{\R}{\mathbb{R}}

\newcommand{\vone}{\mathbf{1}}%
\newcommand{\vzero}{\mathbf{0}}%
\newcommand{\pr}{\mathit{pr}}%

\newcommand{\conv}{\mathrm{conv}}
\newcommand{\NP}{\textsc{NP}}

\newcommand{\PSPACE}{\textsc{PSPACE}}

\newenvironment{qlemma}[1]{%
{\par\mbox{}\newline\noindent\bf Lemma #1.}%
\begin{itshape}%
}{%
\end{itshape}%
}

\newenvironment{qproposition}[1]{%
{\mbox{}\newline\noindent\bf Proposition #1.}
\begin{itshape}%
}{%
\end{itshape}%
}

\newcommand{\M}{\mathcal{M}}%
\newcommand{\rk}{\mathrm{rank}}%
\newcommand{\rkp}{\mathrm{rank}_{+}}%
\newcommand{\rrkp}{\mathrm{rrank}_{+}}%
\newcommand{\back}{\mathop{\mathrm{Back}}}%
\newcommand{\col}{\mathrm{Col}}%
\newcommand{\row}{\mathrm{Row}}%
\newcommand{\ie}{{i.e.}}%
\RequirePackage{tikz}
\usetikzlibrary{arrows,snakes}

\title{On Restricted Nonnegative Matrix Factorization}
\author[1]{Dmitry Chistikov%
    \thanks{Present address: Department of Computer Science, University of Oxford, UK.}}
\author[2]{Stefan Kiefer}
\author[2]{Ines Maru\v{s}i\'{c}}
\author[2]{\mbox{Mahsa Shirmohammadi}}
\author[2]{James Worrell}
\affil[1]{Max Planck Institute for Software Systems (MPI-SWS),
  Germany\\
  \texttt{dch@mpi-sws.org}}
\affil[2]{University of Oxford,
  UK\\
  \texttt{FirstName.LastName@cs.ox.ac.uk}}
\authorrunning{D. Chistikov and S. Kiefer and I. Maru\v{s}i\'{c} and M. Shirmohammadi and J. Worrell} 

\Copyright{Dmitry Chistikov and Stefan Kiefer and Ines Maru\v{s}i\'{c} and  {Mahsa Shirmohammadi} and James Worrell}


\EventEditors{Ioannis Chatzigiannakis, Michael Mitzenmacher, Yuval
Rabani, and Davide Sangiorgi}
\EventNoEds{4}
\EventLongTitle{43rd International Colloquium on Automata, Languages,
and Programming (ICALP 2016)}
\EventShortTitle{ICALP 2016}
\EventAcronym{ICALP}
\EventYear{2016}
\EventDate{July 11--15, 2016}
\EventLocation{Rome, Italy}
\EventLogo{}
\SeriesVolume{55}
\ArticleNo{XXX}

\begin{document}

\makeatletter
\renewcommand*{\@Copyright}{}
\makeatother

\maketitle

\begin{abstract}
 Nonnegative matrix factorization (NMF) is the problem of decomposing
  a given nonnegative $n \times m$ matrix~$M$ into a product of a
  nonnegative $n \times d$ matrix~$W$ and a nonnegative
  $d \times m$ matrix~$H$.  Restricted NMF requires in addition that
  the column spaces of $M$~and~$W$ coincide.  Finding the minimal
  inner dimension~$d$ is known to be \NP-hard, both for NMF and
  restricted NMF.  We show that restricted NMF is closely related to a
  question about the nature of minimal probabilistic automata, posed
  by Paz in his seminal 1971 textbook.  We use this connection to
  answer Paz's question negatively, thus falsifying a positive answer
  claimed in~1974.

  Furthermore, we investigate whether a rational matrix $M$ always has
  a restricted NMF of minimal inner dimension whose
  factors $W$ and $H$ are also rational.  We show that this holds for
  matrices $M$ of rank at most $3$ and we exhibit a rank-$4$ matrix
  for which $W$ and $H$ require irrational entries.
 \end{abstract}


\section{Introduction}
Nonnegative matrix factorization (NMF) is the task of factoring a
matrix of nonnegative real numbers $M$ (henceforth a \emph{nonnegative} matrix) as a product $M= W \cdot H$
such that matrices $W$ and $H$ are also nonnegative.  The smallest inner dimension of any such factorization is called the
\emph{nonnegative rank} of $M$, written $\rkp (M)$.

In machine learning, NMF was popularized by the seminal work of Lee
and Seung~\cite{lee1999learning} as a tool for finding features in
facial-image databases. Since then, NMF has found a broad range of
applications---including document clustering, topic modelling,
computer vision, recommender systems, bioinformatics, and acoustic
signal processing~\cite{BucakG07,BerryGG09,Cichocki09,TjioeBH10,YokotaZCY15,ZhangWFM06}.  
In applications,
matrix $M$ can typically be seen as a matrix of data points: each
column of~$M$ corresponds to a data point and each row to a
feature. Then, computing a nonnegative factorization $M = W \cdot H$
corresponds to expressing the data points (columns of $M$) as convex
combinations of latent factors (columns of $W$), \ie, as linear
combinations of latent factors with nonnegative coefficients (columns
of $H$).

From a computational perspective, perhaps the most 
basic problem concerning NMF is whether a given nonnegative
matrix of rational numbers $M$ admits an NMF with inner dimension at
most a given number~$k$.  Formally, the \emph{NMF problem} asks
whether $\rkp(M) \le k$.  In practical applications, various
heuristics and local-search algorithms are used to compute an
approximate nonnegative factorization, but little is known in terms of
their theoretical guarantees.
The NMF problem under the \emph{separability} assumption of Donoho and Stodden~\cite{DBLP:conf/nips/DonohoS03} is tractable:
an NMF $M = W \cdot H$ is called \emph{separable}
if every column of~$W$ is also a column of~$M$.
In 2012, Arora et al.~\cite{DBLP:conf/stoc/AroraGKM12} showed that it is decidable in polynomial time whether a given matrix admits a separable NMF with a given inner dimension. Further progress was made recently, with several efficient algorithms for
computing near-separable NMFs~\cite{kumar2012fast, doi:10.1137/130940670}.


Vavasis \cite{Vavasis09} showed that the problem of deciding whether
the rank of a nonnegative matrix is equal to its nonnegative rank is
\NP-hard. This result implies that generalizations of this problem,
such as the aforementioned NMF problem, the problem of computing the factors $W, H$ (in both exact and approximate versions), and nonnegative
rank determination, are also \NP-hard. It is not known whether any of
these problems are in~\NP. 

Vavasis~\cite{Vavasis09} notes that the difficulty in proving
membership in \NP\ lies in the fact that a certificate for a positive
answer to the NMF problem seems to require the sought factors: a pair
of nonnegative matrices $W, H$ such that $M = W \cdot H$.  Related to
this, Cohen and Rothblum~\cite{CohenRothblum93} posed the question of
whether, given a nonnegative matrix of rational numbers~$M$, there
always exists an NMF $M= W \cdot H$ of inner dimension equal to
$\rkp (M)$ such that both $W$ and $H$ are also matrices of rational
numbers.  A natural route to proving membership of the NMF problem in~\NP\
would be to give a positive answer to the question of Cohen and
Rothblum (as well as a polynomial bound on the bit-length of the
factors $W$ and~$H$).  However, the question remains open.  Currently
the best complexity bound for the NMF problem is membership
in~\PSPACE, which is obtained by translation into the existential
theory of real-closed fields~\cite{DBLP:conf/stoc/AroraGKM12}.  Such a
translation shows that one can always choose the entries of $W$ and~$H$ to be algebraic numbers.

In this work, we focus on the so-called \emph{restricted} NMF (RNMF)
problem, introduced by Gillis and Glineur~\cite{gillis2012geometric}.
The RNMF problem is defined as the NMF problem, except that the column
spaces of $M$ and~$W$ are required to coincide.  (Note that for any
NMF, the column space of~$M$ is a subspace of the column space
of~$W$.)  This problem has a natural geometric interpretation as the
\emph{nested polytope problem (NPP)}: the problem of finding a
minimum-vertex polytope nested between two given convex polytopes.  In
more detail, for a rank-$r$ matrix~$M$, finding an RNMF with
inner dimension~$d$ is known to correspond exactly to finding a nested polytope
with $d$~vertices in an $(r-1)$-dimensional NPP.

Our contributions are as follows.
\begin{enumerate}
\item We establish a tight connection between NMF and the coverability
  relation in labelled Markov chains (LMCs).  The latter notion was
  introduced by Paz~\cite{paz1971}.  Loosely speaking, an LMC~$\M'$
  \emph{covers} an LMC~$\M$ if for any initial distribution over the
  states of~$\M$ there is an initial distribution over the states
  of~$\M'$ such that $\M$ and~$\M'$ are equivalent.  In 1971,
  Paz~\cite{paz1971} asked a question about the nature of minimal
  covering LMCs.  The question was supposedly answered positively in
  1974~\cite{DBLP:journals/tc/Bancilhon74}.  However, we show that the
  correct answer is negative, thus falsifying the claim
  in~\cite{DBLP:journals/tc/Bancilhon74}.  Instrumental to our
  counterexample is the observation that restricted nonnegative rank
  and nonnegative rank can be different.  (Indeed, the wrong claims
  in~\cite{DBLP:journals/tc/Bancilhon74} seem to implicitly rely on
  the opposite assumption, although the notions of NMF and RNMF had
  not yet been developed.)
\item We show that the RNMF problem for matrices~$M$ of rank~$3$ or
  less can be solved in polynomial time.  In fact, we show that there
  is always a \emph{rational} NMF of~$M$ with inner dimension
  $\rkp(M)$, and that it can be computed in polynomial time in the
  Turing model of computation.  This improves a result
  in~\cite{gillis2012geometric} where the RNMF problem is shown to be
  solvable in polynomial time assuming a RAM model with unit-cost
  arithmetic.  Both our algorithm and the one
  in~\cite{gillis2012geometric} exploit the connection to the
  2-dimensional NPP, allowing us to take advantage of a geometric
  algorithm by Aggarwal et al.~\cite{DBLP:journals/iandc/YapABO89}.
  We need to adapt the algorithm
  in~\cite{DBLP:journals/iandc/YapABO89} to ensure that the occurring
  numbers are rational and can be computed in polynomial time in the
  Turing model of computation.
\item We exhibit a rank-$4$ matrix that has an RNMF with
  inner dimension~$5$ but \emph{no rational} RNMF with
  inner dimension~$5$.  We construct this matrix via a particular
  instance of the 3-dimensional NPP, again taking advantage of the
  geometric interpretation of RNMF.  Our result answers the RNMF
  variant of Cohen and Rothblum's question in~\cite{CohenRothblum93}
  negatively.  The original (NMF) variant remains open.
\end{enumerate}
Detailed proofs of all results can be found
\iftechrep{in the appendix.}{in the full version of this paper.}

\section{Nonnegative Matrix Factorization} \label{sec-prelims}
Let $\mathbb{N}$ and~$\mathbb{N}_{0}$ denote the set of all positive and nonnegative integers, respectively.  For every $n \in \mathbb{N}$, we write $[n]$ for the set $\{1, 2, \ldots, n\}$ and write $I_{n}$ for the identity matrix of order~$n$. For any ordered field $\mathbb{F}$, we denote by $\mathbb{F}_{+}$ the set of all its nonnegative elements. For any vector $v$, we write $v_{i}$ for its $i^\text{th}$ entry. A vector $v$ is called \emph{stochastic} if its entries are nonnegative real numbers that sum up to one. For every $i \in [n]$, we write $e_i$ for the $i^\text{th}$ coordinate vector in $\mathbb{R}^{n}$.
We write $\vone^{(n)}$ for the $n$-dimensional column vector with all ones.
We omit the superscript if it is understood from the context.

For any matrix $M$, we write $M_{i}$ for its $i^\text{th}$ row, $M^{j}$ for its $j^\text{th}$ column, and $M_{i, j}$ for its $(i,j)^\text{th}$ entry. 
The \emph{column space} (resp., \emph{row space}) of $M$, written $\col (M)$ (resp., $\row(M)$), is the vector space spanned by the columns (resp., rows) of $M$.
A matrix is called \emph{nonnegative} (resp., \emph{zero} or \emph{rational}) if so are all its entries.
A nonnegative matrix is \emph{column-stochastic} (resp., \emph{row-stochastic}) if the element sum of each of its columns (resp., rows) is one.

\subsection{Nonnegative Rank}\label{subsec-nonnegRank}

Let $\mathbb{F}$ be an ordered field, such as the reals $\mathbb{R}$ or the rationals $\mathbb{Q}$. Given a nonnegative matrix $M \in \mathbb{F}_+^{n \times m}$, a \emph{nonnegative matrix factorization (NMF) over $\mathbb{F}$} of $M$ is any representation of the form $M = W \cdot H$ where $W \in \mathbb{F}_+^{n \times d}$ and $H \in \mathbb{F}_+^{d \times m}$ are nonnegative matrices.
Note that $\col(M) \subseteq \col(W)$.
We refer to~$d$ as the \emph{inner dimension} of the NMF, and hence refer to NMF $M = W \cdot H$ as being \emph{$d$-dimensional}.
The \emph{nonnegative rank over $\mathbb{F}$} of $M$ is the smallest number $d \in \mathbb{N}_{0}$ such that there exists a $d$-dimensional NMF over $\mathbb{F}$ of $M$. An equivalent characterization~\cite{CohenRothblum93} of the nonnegative rank over $\mathbb{F}$ of $M$ is as the smallest number of rank-$1$ matrices in $\mathbb{F}_{+}^{n \times m}$ such that $M$ is equal to their sum. 
The nonnegative rank over~$\mathbb{R}$ will henceforth simply be called nonnegative rank, and will be denoted by $\rkp(M)$. For any nonnegative matrix $M \in \mathbb{R}_{+}^{n \times m}$, it is easy to see that 
$\rk (M) \le \rkp (M) \le \text{min}\{n, m\}$. 

Given a nonzero matrix $M \in \mathbb{F}_{+}^{n \times m}$, by removing the zero columns of $M$ and dividing each remaining column by the sum of its elements, we obtain a column-stochastic matrix $M^{\prime}$ 
with equal nonnegative rank.
Similarly, if $M = W \cdot H$ then after removing zero columns in~$W$ and multiplying with a suitable diagonal matrix~$D$, we get $M = W \cdot H = W D \cdot D^{-1} H$ where~$W D$ is column-stochastic.
If $M$ is column-stochastic then $\vone^\top = \vone^\top M = \vone^\top W D \cdot D^{-1} H = \vone^\top D^{-1} H$, hence $D^{-1} H$ is column-stochastic as well.
Thus, without loss of generality one can consider NMFs of column-stochastic matrices into column-stochastic matrices~\cite[Theorem 3.2]{CohenRothblum93}.

\begin{quote}\textbf{NMF problem:}
Given a matrix $M \in \mathbb{Q}_{+}^{n \times m}$ and $k \in \N$, is $\rkp(M) \le k$?
\end{quote}

The NMF problem is \NP-hard, even for $k = \rk(M)$ (see~\cite{Vavasis09}).
On the other hand, it is reducible to the existential theory of the reals, hence by~\cite{Can88,Renegar92} it is in \PSPACE.

For a matrix $M \in \mathbb{Q}_{+}^{n \times m}$, its nonnegative rank over $\mathbb{Q}$ is clearly at least $\rkp (M)$.
While those ranks are equal if $\rk (M) \le 2$,
a longstanding open question by Cohen and Rothblum asks whether they are always equal~\cite{CohenRothblum93}.
In other words, it is conceivable that there exists a rational matrix $M \in \mathbb{Q}_{+}^{n \times m}$ with $\rkp(M) = d$ that has no \emph{rational} NMF with inner dimension~$d$. 
Recently, Shitov~\cite{ShitovNonRankSubField15} exhibited a nonnegative matrix (with irrational entries) whose nonnegative rank over a subfield of $\R$ is different from its nonnegative rank over $\R$.

\subsection{Restricted Nonnegative Rank} \label{sec_prelims:RNR}
For all matrices~$M \in \mathbb{F}_{+}^{n \times m}$, 
an NMF $M = W \cdot H$ is called \emph{restricted NMF (RNMF)}~\cite{gillis2012geometric} if $\rk(M) = \rk(W)$.
As we know $\col(M) \subseteq \col(W)$ holds for all NMF instances, the condition
 $\rk(M) = \rk(W)$ is then equivalent to $\col(M) = \col(W)$.
The \emph{restricted nonnegative rank over $\mathbb{F}$} of $M$ is the smallest number $d \in \mathbb{N}_{0}$ such that there exists a $d$-dimensional restricted nonnegative factorization over $\mathbb{F}$ of $M$. 
Unless indicated otherwise, henceforth we will assume $\mathbb{F} = \mathbb{R}$ when speaking of the restricted nonnegative rank of~$M$, and denote it by~$\rrkp(M)$.

\begin{quote}\textbf{RNMF problem:}
Given a matrix $M \in \mathbb{Q}_{+}^{n \times m}$ and $k \in \N$, is \mbox{$\rrkp(M) \le k$}?
\end{quote}

We have the following basic properties. 
\begin{lemma}[\cite{gillis2012geometric}]\label{lemma:basic_properties_RNR}
Let $M \in \mathbb{R}_{+}^{n \times m}$. Then $\rk (M) \le \rkp (M) \le \rrkp (M) \le m$. Moreover, if $\rk (M) = \rkp (M)$ then $\rk (M) = \rrkp (M)$. 
\end{lemma}
Thus, with the above-mentioned \NP-hardness result, it follows that the RNMF problem is also \NP-hard and in \PSPACE.

For a matrix $M \in \mathbb{Q}_{+}^{n \times m}$, its restricted nonnegative rank over $\mathbb{Q}$ is clearly at least $\rrkp (M)$. As with nonnegative rank, in general it is not known whether the restricted nonnegative ranks of $M$ over $\mathbb{R}$ and over $\mathbb{Q}$ are equal.  By~\cite[Theorem 4.1]{CohenRothblum93} and Lemma~\ref{lemma:basic_properties_RNR}, this is true when $\rk (M) \leq 2$.

RNMF has the following geometric interpretation.
For a dimension $\ell \in \N$,
the \emph{convex combination} of a set $\{ v_{1}, \ldots, v_{m} \} \subset \R^\ell$ is a point $\lambda_{1} v_{1} + \cdots + \lambda_{m} v_{m}$
where $(\lambda_{1}, \ldots, \lambda_{m})$ is a stochastic vector. 
The \emph{convex hull} of $\{ v_{1}, \ldots, v_{m} \}$, written as $\conv\{ v_{1}, \ldots, v_{m} \}$, is the set of all convex combinations of $\{ v_{1}, \ldots, v_{m} \}$.
We call $\conv\{ v_{1}, \ldots, v_{m} \}$ a \emph{polytope spanned by} $v_1, \ldots, v_m$.
A \emph{polyhedron} is a set $\{\, x \in \R^\ell \mid A x + b \ge 0 \,\}$ with $A \in \R^{n \times \ell}$ and $b \in \R^n$.
A set is a polytope if and only if it is a bounded polyhedron.
A polytope is \emph{full-dimensional} (i.e., has volume) if the matrix $(A \ \ b) \in \R^{n \times (\ell+1)}$ has rank $\ell+1$.

\begin{quote}\textbf{Nested polytope problem (NPP):}
Given~$r,n \in \N$,
let $A \in \Q^{n \times (r-1)}$ and $b \in \Q^n$ be such that $P = \{\, x \in \R^{r-1} \mid A x + b \ge 0 \,\}$ is a full-dimensional polytope.
Let $S \subseteq P$ be a full-dimensional polytope described by spanning points.
The \emph{nested polytope problem (NPP)} asks, given~$A, b,S$ and a number~$k\in \N$,
 whether there exist $k$~points that span a polytope~$Q$ with $S \subseteq Q \subseteq P$.
Such a polytope~$Q$ is called \emph{nested} \mbox{between $P$~and~$S$}.
\end{quote}

The following proposition appears as Theorem~1 in~\cite{gillis2012geometric}.
\begin{proposition} \label{prop:Belgian-reduction}
 The RNMF problem and the NPP are interreducible in polynomial time.
\end{proposition}
 More specifically, the reductions are as follows.
 \begin{enumerate}
  \item\label{r:nmf-to-npp} Given a nonnegative matrix $M \in \Q_{+}^{n \times m}$ of rank~$r$, one can compute in polynomial time $A \in \Q^{n \times (r-1)}$ and $b \in \Q^n$ such that $P = \{\, x \in \R^{r-1} \mid A x + b \ge 0 \,\}$ is a full-dimensional polytope, and $m$~rational points that span a full-dimensional polytope~$S \subseteq P$ such that
       \begin{enumerate}
         \item[(a)] any $d$-dimensional RNMF (rational or irrational) of~$M$ determines $d$~points that span a polytope~$Q$ with $S \subseteq Q \subseteq P$, and
         \item[(b)] any $d$~points (rational or irrational) that span a polytope~$Q$ with $S \subseteq Q \subseteq P$ determine a $d$-dimensional RNMF of~$M$.
       \end{enumerate}
  \item\label{r:npp-to-nmf} Let $A \in \Q^{n \times (r-1)}$ and $b \in \Q^n$ such that $P = \{\, x \in \R^{r-1} \mid A x + b \ge 0 \,\}$ is a full-dimensional polytope.
      Let $S \subseteq P$ be a full-dimensional polytope spanned by $s_1, \ldots, s_m \in \Q^{r-1}$.
      Then matrix $M \in \Q^{n \times m}$ with $M^i = A s_i + b$ for $i \in [m]$ satisfies (a)~and~(b).
 \end{enumerate}
Importantly, the correspondences (a)~and~(b) preserve rationality.
In \iftechrep{Appendix~\ref{app-prelims}}{the full version} we detail the reduction from point~\ref{r:npp-to-nmf} above, thereby filling in a small gap in the proof of~\cite{gillis2012geometric}.

\begin{example}[{\cite[Example~1]{gillis2012geometric}}] \label{example:rankNeqRrank}
Using the geometric interpretation of restricted nonnegative rank it follows easily that, in general, we may have $\rk(M) < \rkp(M) < \rrkp(M)$. 
Let \emph{3D-cube} NPP be the NPP instance where the inner and outer polytope are the standard 3D cube, i.e., $P = S = \{\, x \in \mathbb{R}^3 \mid x_i \in [0, 1],~ 1 \le i \le 3 \,\}$.
The only nested polytope is $Q = P$. The corresponding restricted NMF problem consists of the following matrix $M \in \mathbb{R}_{+}^{6 \times 8}$:
\[
M = \left(
\begin{smallmatrix}
0 & 0 & 0 & 0 & 1 & 1 & 1 & 1 \\
1 & 1 & 1 & 1 & 0 & 0 & 0 & 0 \\
0 & 0 & 1 & 1 & 0 & 0 & 1 & 1 \\
1 & 1 & 0 & 0 & 1 & 1 & 0 & 0 \\
0 & 1 & 0 & 1 & 0 & 1 & 0 & 1 \\
1 & 0 & 1 & 0 & 1 & 0 & 1 & 0
\end{smallmatrix}
\right).
\] 
We have $\rrkp(M) = 8$ and $\rk(M) = 4$.
Since $\rkp(M)$ is bounded above by the number of rows in~$M$,
we have $\rkp(M) \le 6$.
It is shown in~\cite{gillis2012geometric} that $\rkp(M) = 6$.
\end{example}

\section{Coverability of Labelled Markov Chains} 
\label{sec-coverability} 
In this section, we establish a connection between RNMF and the
coverability relation for labelled Markov chains.  We thereby answer
an open question posed in 1971 by Paz~\cite{paz1971} about the nature
of minimal covering labelled Markov chains.

A \emph{labelled Markov chain} (\emph{LMC}) is a tuple
$\M = (n, \Sigma, \mu)$ where $n \in \N$ is the number of states,
$\Sigma$ is a finite alphabet of labels, and function
$\mu : \Sigma \to [0, 1]^{n \times n }$ specifies the transition
matrices and is such that $\sum_{\sigma \in \Sigma} \mu(\sigma)$ is a
row-stochastic matrix.  The intuitive behaviour of the LMC~$\M$ is as
follows: When $\M$ is in state $i \in [n]$, it emits label $\sigma$
and moves to state $j$, with probability $\mu(\sigma)_{i, j}$.

We extend the function $\mu$ to words by defining
$\mu(\varepsilon) := I_{n}$ and
$\mu(\sigma_{1} \ldots \sigma_{k}) := \mu(\sigma_{1}) \cdots
\mu(\sigma_{k})$ for all~$k \in \N$, and all
$\sigma_{1}, \ldots, \sigma_{k} \in \Sigma$.  Observe that
$\mu( x y) = \mu(x) \cdot \mu(y)$ for all words~$x, y \in \Sigma^{*}$.
We view $\mu(w)$ for a word~$w \in \Sigma^*$ as follows:
if $\M$ is in state $i \in [n]$, it emits~$w$
and moves to state $j$ in $|w|$ steps, with probability
$\mu(w)_{i, j}$.  

For $i \in [n]$ and $w \in \Sigma^*$, we write
$\pr^\M_i(w) := e_i^{\top} \cdot \mu(w) \cdot \vone^{(n)}$ for the
probability that, starting in state $i$, $\M$ emits word $w$.  For
example, in Figure~\ref{fig-LMC} we have
$\pr^\M_0(a_1b_1)=\frac{1}{12}$.  More generally, for a given
\emph{initial distribution}~$\pi$ on the set of states~$[n]$ (viewed
as a stochastic row vector), we write
$\pr^\M_\pi(w) := \pi \cdot \mu(w) \cdot \vone^{(n)}$ for the
probability that $\M$ emits word $w$ starting from state distribution
$\pi$.  \iftechrep{We omit the superscript~$\M$ from $\pr^\M_\pi$ when
  it is clear from the context.}{}


We say that an LMC $\M$ is \emph{covered by} an LMC $\M^{\prime}$, written as $\M^{\prime} \ge \M$, 
if for every initial distribution $\pi$ for~$\M$ there exists a distribution $\pi^{\prime}$ for~$\M'$ such that
	$\pr^{\M}_\pi(w) = \pr^{\M'}_{\pi'}(w)$
for all words~$w \in \Sigma^*$.


The \emph{backward matrix} of~$\M$ is a matrix $\back \M  \in \mathbb{R}_{+}^{[n] \times \Sigma^*}$ where $(\back \M)_{i, w} = \pr^\M_i(w)$ for every $i \in [n]$ and $w \in \Sigma^*$. The \emph{rank} of $\M$ is defined by $\rk(\M) = \rk(\back \M)$.
(Matrix~$\back \M$ is infinite, but since it has $n$~rows, its rank is at most~$n$.)
It follows easily from the definition (see also~\cite[Theorem 3.1]{paz1971}) that $\M^{\prime} \ge \M$ if and only if there exists a row-stochastic matrix $A$ such that $A \cdot \back \M^{\prime}= \back \M$.

LMCs can be seen as a special case of stochastic sequential machines, a class of probabilistic automata introduced and studied by Paz~\cite{paz1971}.
More specifically, they are stochastic sequential machines with a singleton input alphabet and $\Sigma$ as output alphabet.
In his seminal 1971 textbook on probabilistic automata~\cite{paz1971}, Paz asks the following question:

\begin{question}[Paz~\cite{paz1971}, p.~38] \label{question-Paz}
If an $n$-state LMC $\M$ is covered by an $n^{\prime}$-state LMC $\M^{\prime}$ where $n^{\prime} < n$, is $\M$ necessarily covered by some $n^*$-state LMC $\M^*$, where $n^* < n$, such that $\M^*$ and $\M$ have the same rank?
\end{question}
\medskip

\noindent In~1974, a positive answer to this question was claimed~\cite[Theorem~13]{DBLP:journals/tc/Bancilhon74}.
In fact, the author of~\cite{DBLP:journals/tc/Bancilhon74} makes a stronger claim, namely that the answer to Question~\ref{question-Paz} is yes, even if the inequality $n^* < n$ in Question~\ref{question-Paz} is replaced by $n^* \le n'$.
To the contrary, we show:
\begin{theorem} \label{thm-answer-Paz}
The answer to Question~\ref{question-Paz} is negative.
\end{theorem}

\noindent Theorem~\ref{thm-answer-Paz} falsifies the claim in~\cite{DBLP:journals/tc/Bancilhon74}.
In \iftechrep{Appendix~\ref{sub-Hippie-mistakes}}{the full version} we discuss in detail the mistake in~\cite{DBLP:journals/tc/Bancilhon74}.
To prove Theorem~\ref{thm-answer-Paz} we establish a tight connection between NMF and LMC coverability:

\newcommand{\stmtpropLMCdetermine}{
Given a nonnegative matrix $M \in \Q_{+}^{n \times m}$ of rank~$r$, one can compute in polynomial time an LMC $\M = (m+2,\Sigma,\mu)$ of rank~$r+2$ such that for all $d \in \N$:
\begin{enumerate}
\item[(a)] any $d$-dimensional NMF $M = W \cdot H$ determines an LMC $\M' = (d+2,\Sigma,\mu')$ with $\M' \ge \M$ and $\rk(\M') = \rk(W) + 2$, and
\item[(b)] any LMC $\M' = (d+2,\Sigma,\mu')$ with $\M' \ge \M$ determines a $d$-dimensional NMF $M = W \cdot H$  with $\rk(\M') = \rk(W) + 2$.
\end{enumerate}
In particular, for all $d \in \N$ the inequality $\rrkp(M) \le d$ holds if and only 
if $\M$ is covered by some $(d+2)$-state LMC $\M'$ such that $\M'$ and $\M$ have the same rank.
}

\begin{proposition} \label{prop-LMC-determine}
\stmtpropLMCdetermine
\end{proposition}

Assuming Proposition~\ref{prop-LMC-determine} we can prove Theorem~\ref{thm-answer-Paz}:
\begin{proof}[Proof of Theorem~\ref{thm-answer-Paz}]
Let $M \in \{0,1\}^{6 \times 8}$ be the matrix from  Example~\ref{example:rankNeqRrank}. 
Let $\M = (10,\Sigma,\mu)$  be the associated LMC from Proposition~\ref{prop-LMC-determine}.
Since $M = I_6 \cdot M$ is an NMF with inner dimension~$6$, by Proposition~\ref{prop-LMC-determine}~(a) 
there is an LMC~$\M' = (8, \Sigma, \mu')$ with $\M' \ge \M$. 
Towards a contradiction, suppose the answer to Question~\ref{question-Paz} were yes.
Then $\M$ is also covered by some $n^*$-state LMC~$\M^*$, where $n^* \le 9$, such that $\M^*$ and $\M$ have the same rank.
The last sentence of Proposition~\ref{prop-LMC-determine} then implies that $\rrkp(M) \le 7$.
But this contradicts the equality $\rrkp(M) = 8$ from Example~\ref{example:rankNeqRrank}.
Hence, the answer to Question~\ref{question-Paz} is no.
\end{proof}

To prove Proposition~\ref{prop-LMC-determine} we adapt a reduction from NMF to the trace-refinement problem in Markov decision processes~\cite{16FKS-FOSSACS}.

\begin{proof}[Proof sketch of Proposition~\ref{prop-LMC-determine}]
Let $M \in \Q_{+}^{n \times m}$ be a nonnegative matrix of rank~$r$. 
As argued in Section~\ref{subsec-nonnegRank}, without loss of generality we may assume that~$M$ is column-stochastic and consider factorizations of $M$ into column-stochastic matrices only.

We define an LMC $\M = (m+2, \Sigma, \mu)$ with~$m+2$ states~$\{0, 1, \ldots, m, m+1\}$.
The alphabet is $\Sigma = \{a_1, \ldots, a_m\} \cup \{b_1, \ldots, b_n\} \cup \{\checkmark\}$
and the function~$\mu$, for all $i \in [m]$  and all $j \in [n]$, is defined by:
\begin{align*}
	\mu(a_i)_{0, i} &= \textstyle{\frac{1}{m}},
   &\mu(b_j)_{i, m+1} &= (M^{\top})_{i, j} = M_{j, i},
   &\mu(\checkmark)_{m+1, m+1} &= 1, 
\end{align*}
and all other entries of~$\mu(a_i)$, $\mu(b_j)$, and $\mu(\checkmark)$ are~$0$.  See Figure~\ref{fig-LMC} for an example.
We have:
\renewcommand{\kbldelim}{(}
\renewcommand{\kbrdelim}{)}
\vspace{-.2cm}
\[
\back \M
= \!\!\!\!\!
 \kbordermatrix{
    & \varepsilon & b_{1} & \cdots & b_{n} & \checkmark & a_{i}  & a_{i}b_{j} & b_{1}\checkmark  & \cdots & b_{n}\checkmark & \checkmark^2 & \cdots\\
 & 1 & 0 & \cdots & 0 & 0 & \frac{1}{m} & \frac{1}{m} M_{j, i} & 0 & \cdots & 0 & 0 & \cdots\\
 & 1 & M_{1, 1} & \cdots & M_{n, 1} & 0 & 0  & 0 & M_{1, 1} & \cdots & M_{n, 1} & 0 & \cdots\\
 & \vdots & \vdots & \ddots & \vdots & \vdots & \vdots  &  \vdots  & \vdots & \ddots & \vdots & \vdots & \cdots\\
 & 1 & M_{1, m} & \cdots & M_{n, m} & 0 & 0  & 0  & M_{1, m} & \cdots & M_{n, m} & 0 & \cdots\\
 & 1 & 0 & \cdots & 0  & 1 & 0 & 0 & 0 & \cdots & 0  & 1 & \cdots
  }.
\]
Thus $\rk(\M) = \rk(\back \M) \ge \rk(M)+2$. 
The first~$n+2$ columns (indexed by $\varepsilon, b_1, \ldots, b_n, \checkmark$) 
in $\back \M$ span $\col (\back \M)$. Therefore, $\rk(\M) = \rk(M)+2 = r+2$.

\begin{figure}
\begin{minipage}{0.50\textwidth}
\begin{center}
 \scalebox{.9}{\begin{tikzpicture}[xscale=2.5,yscale=2,baseline,>=stealth',every state/.style={minimum size=0.3,inner sep=1}]
\node[state] (0) at (0,0) {$0$};
\node[state] (1) at (1,1) {$1$};
\node[state] (2) at (1,0) {$2$};
\node[state] (3) at (1,-1) {$3$};
\node[state] (4) at (2,0) {$4$};
\node[state,draw=none] (m) at (0,1) {$\M$:};
\draw[->] (0) -- node[above=1mm,pos=0.6] {$\frac13, a_1$} (1);
\draw[->] (0) -- node[above=0mm,pos=0.6] {$\frac13, a_2$} (2);
\draw[->] (0) -- node[below=1mm,pos=0.6] {$\frac13, a_3$} (3);
\draw[->] (1) to[bend left=45]  node[above=1.5mm,pos=0.6] {$\frac14, b_1$} (4);
\draw[->] (1) to[bend right=10] node[above=1mm,pos=0.7] {$\frac34, b_2$} (4);
\draw[->] (2) to[bend left=20]  node[above=-0.5mm,pos=0.3] {$\frac12, b_1$} (4);
\draw[->] (2) to[bend right=20] node[below=-0.7mm,pos=0.3] {$\frac12, b_2$} (4);
\draw[->] (3) to[bend left=10]  node[below=1mm,pos=0.7] {$\frac34, b_1$} (4);
\draw[->] (3) to[bend right=45] node[below=1.5mm,pos=0.6] {$\frac14, b_2$} (4);
\path[->] (4) edge [loop,out=20,in=-20,looseness=13] node[above] {$1, \checkmark$} (4);
\end{tikzpicture}}
\end{center}
\end{minipage}
\begin{minipage}{0.50\textwidth}
\begin{center}
  \scalebox{.9}{\begin{tikzpicture}[xscale=2.5,yscale=2,baseline,>=stealth',every state/.style={minimum size=0.3,inner sep=1}]
\node[state] (0) at (0,0) {$0$};
\node[state] (1) at (1,1) {$1$};
\node[state] (2) at (1,-1) {$2$};
\node[state] (3) at (2,0) {$3$};
\node[state,draw=none] (m) at (-.2,1) {$\M'$:};
\draw[->] (0) to[bend left=45] node[above=1.5mm,pos=0.5] {$\frac{1}{12}, a_1$} (1);
\draw[->] (0) to[bend right=10] node[above=1.5mm,pos=0.4] {$\frac{1}{6}, a_2$} (1);
\draw[->] (0) to[bend right=40] node[right=1mm,pos=0.5] {$\frac{1}{4}, a_3$} (1);
\draw[->] (0) to[bend left =40] node[right=1mm,pos=0.5] {$\frac{1}{4}, a_1$} (2);
\draw[->] (0) to[bend left =10] node[below=0.5mm,pos=0.4] {$\frac{1}{6}, a_2$} (2);
\draw[->] (0) to[bend right=45] node[below=0.5mm,pos=0.5] {$\frac{1}{12}, a_3$} (2);
\draw[->] (1) -- node[above=1mm] {$1, b_1$}(3);
\draw[->] (2) -- node[below=1mm] {$1, b_2$}(3);
\path[->] (3) edge [loop,out=20,in=-20,looseness=13] node[above] {$1, \checkmark$} (3);
\end{tikzpicture}}
	\end{center}
\end{minipage} 
\caption{LMC~$\M$ is constructed from matrix~$M = \big(\begin{smallmatrix} 1/4 && 1/2 && 3/4 \\ 3/4 && 1/2 && 1/4 \end{smallmatrix}\big)$ whereas
 LMC~$\M'$ is obtained by NMF $M = I_2 \cdot M$.}
\label{fig-LMC}
\end{figure}

For $d \in \N$, let $M = W \cdot H$ for some column-stochastic matrices $W \in \mathbb{R}_{+}^{n \times d}$ and $H \in \mathbb{R}_{+}^{d \times m}$. 
Define an LMC $\M^{\prime} = (d+2, \Sigma, \mu^{\prime})$ where
the states are $\{0, 1, \ldots, d, d+1\}$. 
The function  $\mu^{\prime}$, for all $i \in [m]$, $j \in [n]$, and $l \in [d]$, is defined by:
\begin{align*}
	\mu^{\prime}(a_i)_{0, l} &= \textstyle{\frac{1}{m}} H_{l, i},
	&\mu^{\prime}(b_j)_{l, d+1} &= W_{j, l},
	&\mu^{\prime}(\checkmark)_{d+1, d+1} &= 1,
\end{align*}
and all other entries of~$\mu^{\prime}(a_i)$, $\mu^{\prime}(b_j)$, and $\mu^{\prime}(\checkmark)$ are~$0$. 
From the NMF $M = W \cdot H$ it follows that we can factor $\back \M$ as follows:
\vspace{-.2cm}
\renewcommand{\kbldelim}{(}
\renewcommand{\kbrdelim}{)}
\[
\begin{pmatrix}
1 & 0 & \cdots & 0 & 0\\
0 & H_{1, 1} & \cdots & H_{d, 1} & 0\\
\vdots & \vdots & \ddots & \vdots & \vdots\\
0 & H_{1, m} & \cdots & H_{d, m} & 0\\
0 & 0 & \cdots & 0 & 1
\end{pmatrix}
\cdot \!\!\!\!
 \kbordermatrix{
    & \varepsilon & b_{1} & \cdots & b_{n} & \checkmark & a_{i}  & a_{i}b_{j} & b_{j}\checkmark & \checkmark^2 & \cdots\\
        & 1 & 0 & \cdots & 0 & 0 & \frac{1}{m} & \frac{1}{m} M_{j, i} & 0 & 0 & \cdots\\
        & 1 & W_{1, 1} & \cdots & W_{n, 1} & 0 & 0  & 0  & W_{j, 1} & 0 & \cdots\\
        & \vdots & \vdots & \ddots & \vdots & \vdots & \vdots  &  \vdots  & \vdots & \vdots & \cdots\\
        & 1 & W_{1, d} & \cdots & W_{n, d} & 0 & 0  & 0 & W_{j, d} & 0 & \cdots\\
        & 1 & 0 & \cdots & 0 & 1 & 0  & 0  & 0  & 1 & \cdots
  }
\]
where the left factor is row-stochastic (as $H$ is column-stochastic), and the right factor equals $\back \M'$.
It follows that $\M^{\prime} \ge \M$.
\end{proof}

\section{Restricted NMF of Rank-3 Matrices}\label{sec-NPP32}
In this section we consider rational matrices of rank at most $3$.  We
show that for such matrices the restricted nonnegative ranks over $\R$
and~$\Q$ are equal and we give a polynomial-time algorithm that
computes a minimal-dimension RNMF over $\mathbb{Q}$.

\begin{theorem}
  Given a matrix $M\in \mathbb{Q}_{+}^{n \times m}$ where
  $\rk(M) \le 3$, there is a rational RNMF of~$M$ with
  inner dimension $\rrkp(M)$ and it can be computed in polynomial time
  in the Turing model of computation.
\label{thm:polytime}
\end{theorem}

Using  reduction~\ref{r:nmf-to-npp} of Proposition~\ref{prop:Belgian-reduction}, we
can reduce in polynomial time the RNMF problem for rank-$3$
matrices to the 2-dimensional NPP, \ie, the nested polygon
  problem in the plane.  As noted in Section~\ref{sec_prelims:RNR}, the
correspondence between restricted nonnegative factorizations and
nested polygons preserves rationality.  Thus to prove
Theorem~\ref{thm:polytime} it suffices to prove:
\begin{theorem}\label{thm:minRationalNPP2d}
  Given polygons
  $S \subseteq P \subseteq \mathbb{R}^2$ with
  rational vertices, there exists a minimum-vertex
  polygon~$Q$ nested between $P$ and $S$
  that also has rational vertices. Moreover there is an algorithm that,
  given $P$ and $S$, computes such a polygon in
  polynomial time in the Turing machine model.
\end{theorem} 

In fact, Aggarwal et al.~\cite{DBLP:journals/iandc/YapABO89}
give an algorithm for the 2-dimensional NPP and prove that it runs in
polynomial time in the RAM model with unit-cost arithmetic.  However,
they freely use trigonometric functions and do not address the rationality of the output of the algorithm nor
its complexity in the Turing model.  To prove
Theorem~\ref{thm:minRationalNPP2d} we show that, by adopting a
suitable representation of the vertices of a nested polygon, the
algorithm in~\cite{DBLP:journals/iandc/YapABO89} can be adapted so that it runs in
polynomial time in the Turing model.  We furthermore use this
representation to prove that the minimum-vertex nested polygon identified by
the resulting algorithm has rational vertices.

The remainder of the section is devoted to the proof of
Theorem~\ref{thm:minRationalNPP2d}.  We first recall some
terminology from~\cite{DBLP:journals/iandc/YapABO89} and describe
their algorithm.

A \emph{supporting line segment} is a directed line segment, with its
initial and final points on the boundary of the outer polygon $P$,
that touches the inner polygon $S$ on its left.  A nested polygon with
vertices on the boundary of~$P$ is said to be \emph{supporting} if all
but at most one of its edges are supporting line segments.  A polygon
nested between $P$ and $S$ is called \emph{minimal} if it has the
minimum number of vertices among all polygons nested between $P$
and~$S$.  It is shown in~\cite[Lemma 4]{DBLP:journals/iandc/YapABO89}
that there is always a supporting polygon that is also minimal, and
the algorithm given therein outputs such a polygon.

Let $k$ be the number of vertices of a minimal nested polygon.
Given a vertex $v$ on the boundary of $P$, there is a
uniquely defined supporting polygon $Q_v$ with at most $k+1$
vertices. To determine $Q_v$ one computes the supporting
line segments $v_1v_2,\ldots,v_kv_{k+1}$, where $v_1=v$; see Figure~\ref{fig-aggarwal}.
Then $Q_v$ is either the polygon with vertices
$v_1,\ldots,v_k$ or the polygon with vertices $v_1,\ldots,v_{k+1}$.
In the first case, $Q_v$ is minimal.  The idea behind the algorithm
of~\cite{DBLP:journals/iandc/YapABO89} is to search along the
boundary of $P$ for an initial vertex $v$ such that
$Q_v$ is minimal.

As a central ingredient for our proof of Theorem~\ref{thm:minRationalNPP2d}, we choose a convenient
representation of the vertices of supporting polygons.  To this end, we
assume that the edges of~$P$ are oriented counter-clockwise, and we
represent a vertex $v$ on an edge $pq$ of~$P$ by the unique
$\lambda \in [0,1]$ such that $v = (1-\lambda) p+ \lambda q$.  We call
this the \emph{convex representation} of $v$.

Similar to~\cite{DBLP:journals/iandc/YapABO89}, we
associate with each supporting line segment $uv$ a \emph{ray
  function}~$r$, such that if $\lambda$ is the convex representation
of $u$ then $r(\lambda)$ is the convex representation of $v$.  The
same ray function applies for supporting line segments $u'v'$ with $u'$ in a
small enough interval containing $u$.

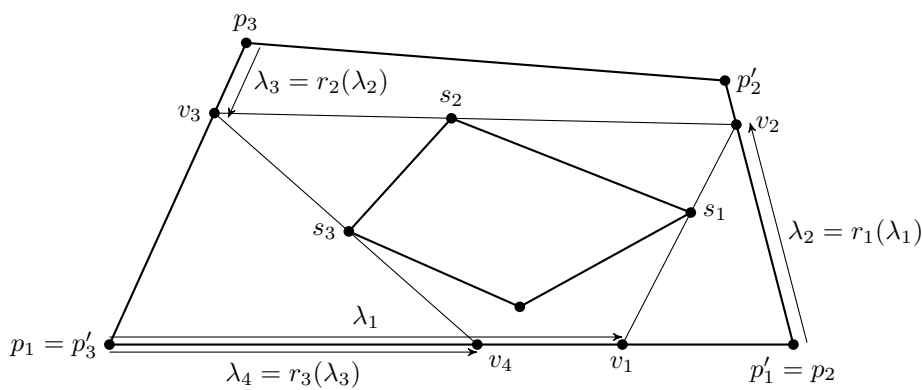
\begin{figure}
\begin{center}
\begin{tikzpicture}[xscale=9,yscale=5,
dot/.style={circle,fill=black,minimum size=4pt,inner sep=0pt,outer sep=-1pt},
]
\coordinate (P1) at (0,0);
\coordinate (P2) at (1,0);
\coordinate (P2p) at (0.9,0.7);
\coordinate (P3) at (0.2, 0.8);
\coordinate (Q1) at (0.85, 0.35);
\coordinate (Q2) at (0.5, 0.6);
\coordinate (Q3) at (0.35, 0.3);
\coordinate (Q4) at (0.6, 0.1);
\coordinate (V1) at (0.75,0);
\coordinate (V2) at  (intersection cs:
       first line={(P2)--(P2p)},
       second line={(V1)--(Q1)});
\coordinate (V3) at  (intersection cs:
       first line={(P3)--(P1)},
       second line={(V2)--(Q2)});
\coordinate (V4) at  (intersection cs:
       first line={(P1)--(P2)},
       second line={(V3)--(Q3)});

\draw[thick] (P1) -- (P2) -- (P2p) -- (P3) -- cycle;
\draw[thick] (Q1) -- (Q2) -- (Q3) -- (Q4) -- cycle;
\draw (V1) -- (V2) -- (V3) -- (V4);
\node[dot,label={[label distance=0mm]180:$p_1=p_3'$}] at (P1) {};
\node[dot,label={[label distance=0mm]-90:$p_1'=p_2$}] at (P2) {};
\node[dot,label={[label distance=0mm]0:$p_2'$}] at (P2p) {};
\node[dot,label={[label distance=0mm]90:$p_3$}] at (P3) {};
\node[dot,label={[label distance=0mm]0:$s_1$}] at (Q1) {};
\node[dot,label={[label distance=0mm]90:$s_2$}] at (Q2) {};
\node[dot,label={[label distance=0mm]180:$s_3$}] at (Q3) {};
\node[dot] at (Q4) {};
\node[dot,label={[label distance=0mm]-90:$v_1$}] at (V1) {};
\node[dot,label={[label distance=1mm]0:$v_2$}] at (V2) {};
\node[dot,label={[label distance=0mm]180:$v_3$}] at (V3) {};
\node[dot,label={[label distance=0mm]-45:$v_4$}] at (V4) {};
\draw[->,>=stealth'] ($ (P1)+(0,0.02) $) -- ($ (V1)+(0,0.02) $) node[midway,above] {$\lambda_1$};
\draw[->,>=stealth'] ($ (P1)-(0,0.02) $) -- ($ (V4)-(0,0.02) $) node[midway,below] {$\lambda_4 = r_3(\lambda_3)$};
\draw[->,>=stealth'] ($ (P2)+(0.02,0.005) $) -- ($ (V2)+(0.02,0.005) $) node[midway,right] {$\lambda_2 = r_1(\lambda_1)$};
\draw[->,>=stealth'] ($ (P3)+(0.02,-0.012) $) -- ($ (V3)+(0.02,-0.012) $) node[midway,right] {$\lambda_3 = r_2(\lambda_2)$};
\end{tikzpicture}
\end{center}
\caption{Supporting polygon $Q_{v_{1}}$. For every $i \in [3]$, vertex $v_{i}$ lies on edge $p_{i} p_{i}^{\prime}$ of $P$, and $s_{i}$ is the point where the supporting line segment $v_{i} v_{i+1}$ touches the inner polygon $S$ on its left.}
\label{fig-aggarwal}
\end{figure}

\newcommand{\stmtlemrationallinfract}{
 Consider bounded polygons $S \subseteq P \subseteq \mathbb{R}^2$
  whose vertices are rational and of bit-length $L$.  Then the ray
  function associated with a supporting line segment $uv$ has the form
  $r(\lambda) = \frac{a\lambda + b}{c\lambda + d}$, where coefficients
  $a,b,c,d$ are rational numbers with bit-length $O(L)$ that can be computed in polynomial time.
}
In the following, when we say polynomial time, we mean polynomial time in the Turing model.
To obtain a polynomial time bound, the key lemma is as follows:
\begin{lemma}\label{lem:rationallinfract}
 \stmtlemrationallinfract
\end{lemma}

Suppose that $v_1v_2,\ldots,v_kv_{k+1}$ is a sequence of $k$
supporting line segments, with corresponding ray functions $r_1,\ldots,r_k$.
Then $v_1,\ldots,v_k$ are the vertices of a minimal supporting polygon
if and only if $(r_k \circ \ldots \circ r_1)(\lambda) \geq \lambda$,
where $\lambda$ is the convex representation of $v_1$.

It follows from~\cite{DBLP:journals/iandc/YapABO89} that, for each edge of~$P$, one can compute in polynomial time a partition~$\mathcal{I}$ of $[0,1]$ into intervals with rational endpoints such that if $\lambda_1, \lambda_2$ are in the same interval $I \in \mathcal{I}$ then the points with convex representation $\lambda_1$ and~$\lambda_2$ are associated with the same sequence of ray functions $r_1, \ldots, r_k$. Using Lemma~\ref{lem:rationallinfract} we can, for each interval $I \in \mathcal{I}$, compute these ray functions in polynomial time.
Define the \emph{slack function} \label{pageref-slack} $s(\lambda)= (r_k \circ \ldots \circ r_1)(\lambda)-\lambda$.
In fact, this function has the form
$s(\lambda)=\frac{a\lambda+b}{c\lambda+d}-\lambda$ for rational numbers $a,b,c,d$ that are also computable in polynomial time\iftechrep{ (see Corollary~\ref{cor:bitsizeNestedVertices} in Appendix~\ref{app-NPP2d})}{}.
Then it is straightforward to check whether $s(\lambda) \ge 0$ has a solution $\lambda \in I$.

Next we show that if such a solution exists, then there exists a rational solution, which, moreover, can be computed in polynomial time. To this end, let $\lambda^* \in I$ be such that $s(\lambda^*) \ge 0$. If $\lambda^*$ is on the boundary of~$I$, then $\lambda^* \in \Q$. If $\lambda^*$ is not on the boundary and is not an isolated solution, then there exists a rational solution in its neighbourhood. Lastly, let $\lambda^*$ be an isolated solution not on the boundary.
Then, $\lambda^*$ is a root of both $s$ and its derivative~$s'$. For every $\lambda \in I$, we have
\begin{align*}
(c\lambda + d) \cdot s(\lambda) =  a\lambda + b - \lambda \cdot (c\lambda + d).
\end{align*} 
Taking the derivative of the above equation with respect to $\lambda$, we get
\begin{align}
c \cdot s(\lambda)  + (c\lambda + d) \cdot s^{\prime} (\lambda) =  a - d - 2c\lambda. \label{eq:slackderivative}
\end{align} 
Since $s(\lambda^*) = s'(\lambda^*) = 0$, from (\ref{eq:slackderivative}) we get $0 = a - d - 2c \lambda^*$.
Note that $c \neq 0$ since otherwise $s \equiv 0$. Therefore,
$\lambda^* = \frac{a-d}{2c} \in \mathbb{Q}$.

It follows that the vertex~$v$ represented by~$\lambda^*$ has rational coordinates computable in polynomial time.
By computing $(r_i \circ \ldots \circ r_1)(\lambda^*)$ for $i \in [k]$, we can compute in polynomial time the convex representation of all vertices of the supporting polygon~$Q_v$.
Observe, in particular, that all vertices are rational.
Hence we have proved Theorem~\ref{thm:minRationalNPP2d}.

\newcommand{\objectThreeD}[4]{
\pgfmathsetmacro{\phiy}{#1} 
\pgfmathsetmacro{\phix}{#2} 
\pgfmathsetmacro{\phiz}{#3} 
\pgfmathsetmacro{\xa}{cos(\phiz)*cos(\phiy)-sin(\phiz)*cos(\phix)*sin(\phiy)}
\pgfmathsetmacro{\xb}{-sin(\phix)*sin(\phiy)}
\pgfmathsetmacro{\ya}{-sin(\phiz)*sin(\phix)}
\pgfmathsetmacro{\yb}{cos(\phix)}
\pgfmathsetmacro{\za}{-cos(\phiz)*sin(\phiy)-sin(\phiz)*cos(\phix)*cos(\phiy)}
\pgfmathsetmacro{\zb}{-sin(\phix)*cos(\phiy)}
\begin{tikzpicture}[x  = {(\xa cm,\xb cm)},
                    y  = {(\ya cm,\yb cm)},
                    z  = {(\za cm,\zb cm)},
                    scale = #4,
                    dot/.style={circle,fill=black,minimum size=4pt,inner sep=0pt,outer sep=-1pt},
]
\draw[line join=round,thick,fill=brown!90] (0,0,0) coordinate (O1) -- (1,0,0) coordinate (O2) -- (1,1/2,0) coordinate (O3) -- (0,1,0) coordinate (O4) -- cycle;
\draw[line join=round,thick,fill=blue!50]  (O1) -- (O2) -- (9/4,0,1/2) coordinate (O5) -- (0,0,8/7) coordinate (O6) -- cycle;

\draw[thick,->, >=angle 60] (O1) -- (O2) node[below] {$x$};
\draw[thick,->, >=angle 45] (O1) -- (O4) node[left] {$y$};
\draw[thick,->, >=angle 90] (O1) -- (O6) node[left] {$z$};

\coordinate (I1) at (3/4,1/8,0);
\coordinate (I2) at (3/4,1/2,0);
\coordinate (I3) at (3/11,17/22,0);
\coordinate (I4) at (2,0,1/2);
\coordinate (I5) at (1/2,0,3/4);
\coordinate (I6) at (1/6,0,7/12);
\foreach \x in {1,2,...,6} \node[dot] at (I\x) {};

\coordinate (M1) at ({2-sqrt(2)}, 0,0);
\coordinate (M2) at  (intersection cs:
       first line={(M1)--(I1)},
       second line={(O2)--(O3)});
\coordinate (M3) at  (intersection cs:
       first line={(M2)--(I2)},
       second line={(O3)--(O4)});
\coordinate (M4) at  (intersection cs:
       first line={(M1)--(I4)},
       second line={(O5)--(O6)});
\coordinate (M5) at  (intersection cs:
       first line={(M4)--(I5)},
       second line={(O1)--(O6)});

\draw (M1) -- (M2) -- (M3) -- cycle;
\draw (M1) -- (M4) -- (M5) -- cycle;

\draw[line join=round,thick,fill=yellow!80]  (O2) -- (O3) -- (O5) -- cycle;
\draw[line join=round,thick,fill=red!40]  (O1) -- (O4) -- (O6) -- cycle;
\draw[line join=round,thick,fill=green!90,opacity=0.4, draw opacity=1]  (O3) -- (O4) -- (O5) -- cycle;
\draw[line join=round,thick,fill=orange!60,opacity=0.35, draw opacity=1]  (O4) -- (O5) -- (O6) -- cycle;
\end{tikzpicture}
}

\section{Restricted NMF Requires Irrationality} \label{sec-NPP3d}

Here we show that the restricted nonnegative ranks over $\R$ and~$\Q$ are, in general, different.
\begin{theorem} \label{thm-NPP3d}
Let
\[
M = 
\begin{pmatrix}
1/8   & 1/2   & 17/22 & 0    & 0     & 0     \\
0     & 0     & 0     & 1/2  & 3/4   & 7/12  \\
3/4   & 3/4   & 3/11  & 2    & 1/2   & 1/6   \\
1/4   & 1/4   & 8/11  & 1/4  & 19/8  & 55/24 \\
1/2   & 1/8   & 1/11  & 1/8  & 15/16 & 17/16 \\
11/16 & 5/16  & 7/44  & 1/16 & 7/32  & 43/96
\end{pmatrix}
\in \Q_+^{6 \times 6}\,.
\]
The restricted nonnegative rank of~$M$ over~$\R$ is~$5$.
The restricted nonnegative rank of~$M$ over~$\Q$ is~$6$.
\end{theorem}
\begin{proof}
Matrix~$M$ has an NMF $M = W \cdot H$ with inner dimension~$5$ with $W, H$ as follows:
\begin{align*}
W
&= 
\begin{pmatrix}
0 & \frac{3 + \sqrt{2}}{14} & \frac{11 + \sqrt{2}}{14} & 0 & 0 \\
0 & 0 & 0 & \frac{12 - 2 \sqrt{2}}{17} & \frac{5}{7} + \frac{\sqrt{2}}{14} \\
2 - \sqrt{2} & 1 & \frac{3 - \sqrt{2}}{7} & \frac{26 + 7 \sqrt{2}}{17} & 0 \\
-1 + \sqrt{2} & 0 & \frac{4 + \sqrt{2}}{7} & \frac{21 - 12 \sqrt{2}}{17} & \frac{39}{14} + \frac{5 \sqrt{2}}{28} \\
\frac{\sqrt{2}}{2} & \frac{2}{7} - \frac{\sqrt{2}}{14} & 0 & \frac{7 - 4 \sqrt{2}}{17} & \frac{33}{28} + \frac{\sqrt{2}}{56} \\
\frac{1}{2} + \frac{\sqrt{2}}{4} & \frac{15}{28} - \frac{\sqrt{2}}{14} & \frac{3 - \sqrt{2}}{28} & 0 & \frac{3}{8} - \frac{\sqrt{2}}{16}
\end{pmatrix}, \\
H
&=
\begin{pmatrix}
\frac{1+\sqrt{2}}{4} & 0 & \frac{\sqrt{2}}{11} & \frac{1}{4} - \frac{\sqrt{2}}{8} & 0 & \frac{1}{6} + \frac{\sqrt{2}}{12} \\
\frac{3 - \sqrt{2}}{4} & \frac{1}{2} + \frac{\sqrt{2}}{8} & 0 & 0 & 0 & 0 \\
0 & \frac{1}{2} - \frac{\sqrt{2}}{8} & 1 - \frac{\sqrt{2}}{11} & 0 & 0 & 0 \\
0 & 0 & 0 & \frac{3}{4} + \frac{\sqrt{2}}{8} & \frac{13}{34} - \frac{7 \sqrt{2}}{68} & 0 \\
0 & 0 & 0 & 0 & \frac{21}{34} + \frac{7 \sqrt{2}}{68} & \frac{5}{6} - \frac{\sqrt{2}}{12}
\end{pmatrix}.
\end{align*}
Since $\rk(M) = \rk(W) = 4$, the NMF $M = W \cdot H$ is restricted.
This RNMF has been obtained by reducing, according to Proposition~\ref{prop:Belgian-reduction}, an NPP instance, which we now describe.

\begin{figure}
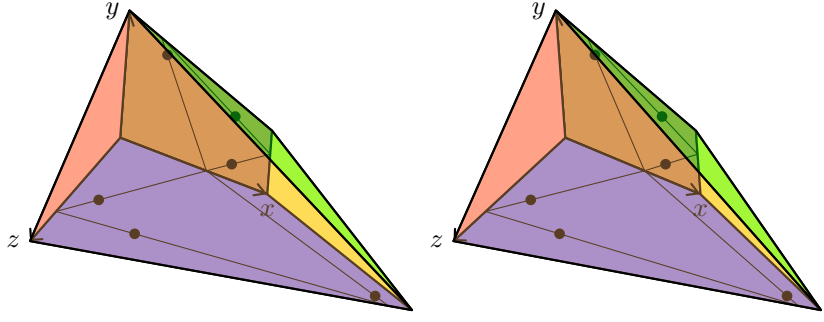

\begin{center}
\objectThreeD{32}{40}{-5}{2.2}\hspace{-13mm}
\objectThreeD{32}{40}{5}{2.2}
\end{center}
\caption{Instance of the nested polytope problem. The two images show orthogonal projections of a 3-dimensional  outer polytope~$P$. The black dots indicate 6 inner points (3 on the brown $xy$-face, and 3 on the blue $xz$-face) that span the interior polytope~$S$. The two triangles on the $xy$-face and on the $xz$-face indicate the (unique) location of 5~points that span the nested polytope~$Q$.
The two slightly different projections are designed to create a 3-dimensional impression using stereoscopy. The ``parallel-eye'' technique should be used, see, e.g., \cite{DSimanekParallelEye}.
See \iftechrep{Figure~\ref{fig-3d-cross-eyed} in Appendix~\ref{app-NPP3d}}{the full version} for a  ``cross-eyed'' variant.
}
\label{fig-3d-parallel-eyed}
\end{figure} 

Figure~\ref{fig-3d-parallel-eyed} shows the outer 3-dimensional polytope~$P$ with 6 faces.
The polytope~$P$ is the intersection of the following half-spaces:
$y \ge 0$ (blue), 
$z \ge 0$ (brown),
$x \ge 0$ (pink),
$-x + \frac{5}{2} z + 1 \ge 0$ (yellow),
$-\frac{1}{2} x - y + \frac{1}{4} z + 1 \ge 0$ (green),
$-\frac{1}{4} x - y - \frac{7}{8} z + 1 \ge 0$ (transparent front).
The figure also indicates an interior polytope~$S$ spanned by 6 points (black dots):
$s_1 = (\frac{3}{4}, \frac{1}{8}, 0)^\top$, 
$s_2 = (\frac{3}{4}, \frac{1}{2}, 0)^\top$,
$s_3 = (\frac{3}{11}, \frac{17}{22},0)^\top$,
$s_4 = (2,0,\frac{1}{2})^\top$, 
$s_5 = (\frac{1}{2},0,\frac{3}{4})^\top$, 
$s_6 = (\frac{1}{6},0,\frac{7}{12})^\top$.
In the following we discuss possible locations of 5 points $q_1, q_2, q_3, q_4, q_5$ that span a nested polytope~$Q$.
Since $s_1, s_2, s_3$ all lie on the (brown) face on the $xy$-plane, but not on a common line, at least 3 of the~$q_i$ must lie on the $xy$-plane.
A similar statement holds for $s_4, s_5, s_6$ and the $xz$-plane.
So at least one~$q_i$, say~$q_1$, must lie on the $x$-axis.

Suppose another~$q_i$, say~$q_2$, lies on the $x$-axis.
Without loss of generality we can take $q_1 = (0,0,0)^\top$ and $q_2 = (1,0,0)^\top$, as all points in~$P$ on the $x$-axis are enclosed by these $q_1$, $q_2$.
%
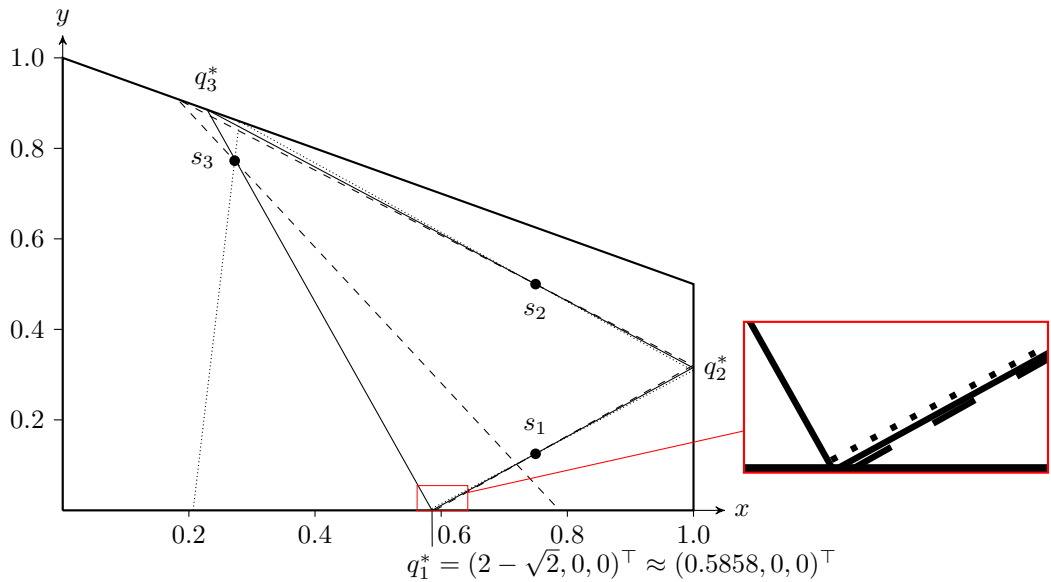
\begin{figure}
\begin{center}
\begin{tikzpicture}[
xscale=8.3,
yscale=6,
dot/.style={circle,fill=black,minimum size=4pt,inner sep=0pt,outer sep=-1pt},
spy using outlines={rectangle, magnification=6, width=4cm, height=2cm, connect spies}
]
\draw[->, >=stealth'] (0,0) -- (1.05,0) node[right] {$x$};
\draw[->, >=stealth'] (0,0) -- (0,1.05) node[above] {$y$};
\foreach \x in {0.2,0.4,0.8,1.0}
    \draw[shift={(\x,0)}] (0pt,0pt) -- (0pt,-0.4pt) node[below] {$\x$};
\draw[shift={(0.6,0)}] (0pt,0pt) -- (0pt,-0.4pt) node[below,xshift=5pt] {$0.6$};
\foreach \x in {0.2,0.4,0.6,0.8,1.0}
    \draw[shift={(0,\x)}] (0,0) -- (-0.4pt,0) node[left] {$\x$};
\draw[thick] (0,0) coordinate (O1) -- (1,0) coordinate (O2) -- (1,1/2) coordinate (O3) -- (0,1) coordinate (O4) -- (O1);
\coordinate (I1) at (3/4,1/8);
\coordinate (I2) at (3/4,1/2);
\coordinate (I3) at (3/11,17/22);
\node[dot,label={[label distance=1mm]90:$s_1$}] at (I1) {};
\node[dot,label={[label distance=1mm]-90:$s_2$}] at (I2) {};
\node[dot,label={[label distance=1mm]180:$s_3$}] at (I3) {};
\coordinate (M1) at ({2-sqrt(2)}, 0);
\coordinate (M2) at  (intersection cs:
       first line={(M1)--(I1)},
       second line={(O2)--(O3)});
\coordinate (M3) at  (intersection cs:
       first line={(M2)--(I2)},
       second line={(O3)--(O4)});
\node[label={[label distance=-1mm]0:$q_2^*$}] at (M2) {};
\node[label={90:$q_3^*$}] at (M3) {};
\draw (M1) -- (M2) -- (M3) -- (M1);
\draw (M1) -- ({2-sqrt(2)}, -0.08) node[below,xshift=25mm,yshift=1mm] {$q_1^* = (2-\sqrt{2},0,0)^\top \approx (0.5858,0,0)^\top$};

\coordinate (M1+) at (0.59, 0);
\coordinate (M2+) at  (intersection cs:
       first line={(M1+)--(I1)},
       second line={(O2)--(O3)});
\coordinate (M3+) at  (intersection cs:
       first line={(M2+)--(I2)},
       second line={(O3)--(O4)});
\coordinate (M4+) at  (intersection cs:
       first line={(M3+)--(I3)},
       second line={(O1)--(O2)});
\draw[dashed] (M1+) -- (M2+) -- (M3+) -- (M4+);

\coordinate (M1-) at (0.58, 0);
\coordinate (M2-) at  (intersection cs:
       first line={(M1-)--(I1)},
       second line={(O2)--(O3)});
\coordinate (M3-) at  (intersection cs:
       first line={(M2-)--(I2)},
       second line={(O3)--(O4)});
\coordinate (M4-) at  (intersection cs:
       first line={(M3-)--(I3)},
       second line={(O1)--(O2)});
\draw[densely dotted] (M1-) -- (M2-) -- (M3-) -- (M4-);

\coordinate (M1spy) at (0.602,0.027);
\spy [red] on (M1spy) in node [right] at (1.08,0.25);
\end{tikzpicture}
\end{center}
\caption{Detailed view of the $xy$-plane.
The outer quadrilateral is one of 6 faces of~$P$, the brown face in Figure~\ref{fig-3d-parallel-eyed}.
The points $s_1, s_2, s_3$ are among the 6 points that span the inner polytope~$S$.
The points $q_1^*, q_2^*, q_3^*$ are among the 5 points that span the nested polytope~$Q$.
The area around~$q_1^*$ is zoomed in on the right-hand side.
The picture illustrates that $q_1^*$ cannot be moved left on the $x$-axis without increasing the number of vertices of the nested polytope:
A dotted ray from a point slightly to the left of~$q_1^*$ is drawn through~$s_1$.
Its intersection with the line $x = 1$ is slightly below~$q_2^*$.
Following the algorithm of~\cite{DBLP:journals/iandc/YapABO89}, the dotted ray is continued in a similar fashion, ``wrapping around'' $s_2$~and~$s_3$, and ending on the $x$-axis at around $x \approx 0.2$, far left of the starting point.
On the other hand, the dashed line illustrates that $q_1^*$ could be moved right (considering only this face).
}
\label{fig-concave-display}
\end{figure}
Figure~\ref{fig-concave-display} provides a detailed view of the $xy$-plane.
To enclose~$s_2$, some $q \in \{q_3, q_4, q_5\}$ must also lie on the $xy$-plane and  to the right of the line that connects $q_2 = (1,0,0)^\top$ and~$s_2$.
To enclose~$s_3$, some $q' \in \{q_3, q_4, q_5\}$ must also lie on the $xy$-plane and  to the left of the line that connects $q_1 = (0,0,0)^\top$ and~$s_3$.
If $q$ and~$q'$ were identical then they would lie outside~$P$---a contradiction.
Hence 4 points (namely, $q_1, q_2, q, q'$) are on the $xy$-plane.
This leaves only one point, say~$q''$, that is not on the $xy$-plane.
To enclose~$s_4$ (see \iftechrep{Figure~\ref{fig-convex-display} in Appendix~\ref{app-NPP3d}}{the corresponding figure in the full version}), point~$q''$ must lie on the $xz$-plane and must lie to the right of the line that connects $q_2 = (1,0,0)^\top$ and~$s_4$.
To enclose~$s_6$, point~$q''$ must lie to the left of the line that connects $q_1 = (0,0,0)^\top$ and~$s_6$.
Hence $q''$ lies outside~$P$---a contradiction.

Hence we have shown that only one point, say~$q_1$, lies on the $x$-axis, and two points besides~$q_1$, say $q_2, q_3$, lie on the $xy$-plane, and two points besides~$q_1$, say $q_4, q_5$, lie on the $xz$-plane.
Figure~\ref{fig-concave-display} indicates a possible location ($q_1^*, q_2^*, q_3^*$) of $q_1, q_2, q_3$.
The figure illustrates that the $x$-coordinate of~$q_1^*$ must be at least $2-\sqrt{2}$.
\begin{figure}
\begin{center}
\begin{tikzpicture}[xscale=700,yscale=7]
\draw[->, >=stealth'] (0.579,0) -- (0.591,0) node[right] {$\lambda$};
\foreach \x in {0.58,0.582,0.584,0.588,0.59}
    \draw[shift={(\x,0)}] (0pt,0pt) -- (0pt,-0.4pt) node[below] {$\x$};
\draw[shift={(0.586,0)}] (0pt,0pt) -- (0pt,-0.4pt) node[below,xshift=7pt] {$0.586$};
\draw[->, >=stealth'] (0.579,-0.35) -- (0.579,0.35) node[above] {}; 
\foreach \x in {-0.3,-0.2,-0.1,0,0.1,0.2,0.3}
    \draw[shift={(0.579,\x)}] (0,0) -- (-0.004pt,0) node[left] {$\x$};
\draw[smooth,samples=94,domain=0.58:0.59]
   plot(\x,{-15*(-4*\x+2+\x*\x)/(15*\x-8)});
\node[left] at (0.5883,0.15) {$s(\lambda)$};
\draw ({2-sqrt(2)},0) -- ({2-sqrt(2)}, -4pt) node[below] {$2-\sqrt{2} \approx 0.5858$};
\end{tikzpicture}
\end{center}
\caption{The slack function $s(\lambda) = \frac{52 \lambda - 30}{15 \lambda - 8} - \lambda$ corresponding to Figure~\ref{fig-concave-display}.
When $s(\lambda)<0$, there is no nested triangle with vertex $(\lambda,0,0)$.
}
\label{fig-concave-slack}
\end{figure}
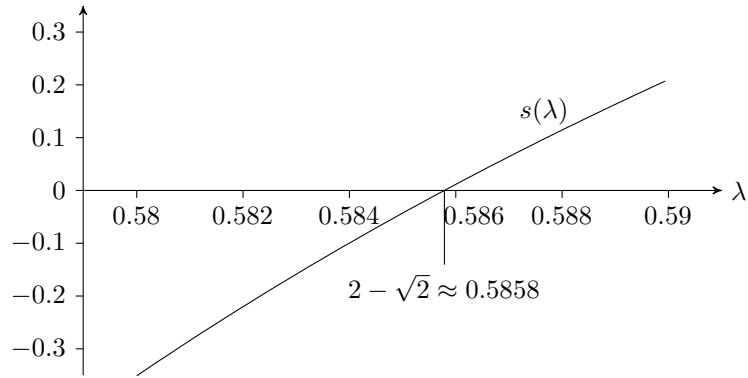 
Figure~\ref{fig-concave-slack} illustrates how to prove the same fact more formally, using the concept of a slack function (see 
Section~\ref{sec-NPP32}):
The slack function~$s(\lambda)$ for the interval containing $2-\sqrt{2}$ has a zero at $\lambda = 2 - \sqrt{2}$, with a sign change from negative to positive.
An inspection of the intervals (of the partition~$\mathcal{I}$ from Section~\ref{sec-NPP32}) to the ``left'' of $2 - \sqrt{2}$ reveals that none of the corresponding slack functions~$\tilde{s}$ satisfies $\tilde{s}(\lambda) \ge 0$ for $\lambda < 2 - \sqrt{2}$.
Similarly, 
the $x$-coordinate of~$q_1^*$ must be \emph{at most} $2-\sqrt{2}$,
\iftechrep{as illustrated by Figures~\ref{fig-convex-display}~and~\ref{fig-convex-slack} in Appendix~\ref{app-NPP3d}}{see corresponding figures in the full version}.
Hence $q_1^* = (2-\sqrt{2}, 0, 0)^\top$ is necessary.
This uniquely (up to permutations) determines $q_2^*, q_3^*$ and similarly the locations $q_4^*, q_5^*$ of $q_4, q_5$.
With the reduction from Proposition~\ref{prop:Belgian-reduction} this NPP solution determines the RNMF of~$M$ mentioned at the beginning of the proof.
Since there is no $4$-point solution of the NPP instance, we have $\rrkp(M) = 5$.
(Since $\rk(M) = 4$, Lemma~\ref{lemma:basic_properties_RNR} implies $\rkp(M) = 5$.)
Since there is no $5$-point rational solution of the NPP instance, the restricted nonnegative rank of~$M$ over~$\Q$ is~$6$.
\end{proof}

\section{Conclusion and Future Work}

We have shown that an optimal \emph{restricted} nonnegative factorization of a rational matrix may require factors that have irrational entries. An outstanding open problem is whether the same holds for general nonnegative factorizations.  An answer to this question will likely shed light on the issue of whether the nonnegative rank can be computed in \NP.  

Another contribution of the paper has been to develop connections between nonnegative matrix factorization and probabilistic automata, thereby
answering an old question
concerning the latter. Pursuing this connection, and closely related to the above-mentioned open problem, one can ask whether, given a probabilistic automaton with rational transition probabilities, one can always find a minimal equivalent probabilistic automaton that also has rational transition probabilities.

\subparagraph*{Acknowledgements.}

The authors would like to thank Michael Benedikt for stimulating discussions, and anonymous referees for their helpful suggestions.
Chistikov was sponsored in part by the ERC Synergy award ImPACT.
Kiefer is supported by a University Research Fellowship of the Royal Society.
Kiefer, Maru\v{s}i\'{c}, Shirmohammadi, and Worrell gratefully acknowledge the support of the EPSRC.

\bibliographystyle{plainurl}
\bibliography{references}

\iftechrep{
\newpage
\appendix
\section{Proofs of Section~\ref{sec-prelims}} \label{app-prelims}

\newcommand{\augA}{(A \ \ b)}%
\newcommand{\chat}{\hat c}%
We show point~\ref{r:npp-to-nmf} from the description of Proposition~\ref{prop:Belgian-reduction}.
Let $A \in \Q^{n \times (r-1)}$ and $b \in \Q^n$ such that $P = \{\, x \in \R^{r-1} \mid A x + b \ge 0 \,\}$ is a full-dimensional polytope.
Hence $\augA \in \Q^{n \times r}$ has full rank~$r$.
Let $S = \conv\{s_1, \ldots, s_m\} \subset P$ be a full-dimensional polytope.
Define matrix $M \in \Q^{n \times m}$ with $M^j = A s_j + b$ for $j \in [m]$.
We show the correspondences (a)~and~(b) from the main text.
\begin{enumerate}
\item[(a)]
Consider an RNMF $M = W \cdot H$ with inner dimension~$d$.
We can assume that $W$ has no zero-columns.
We have $\col(W) = \col(M) = \col(\augA)$: indeed, the first equality holds because the NMF is restricted, and the second equality holds as $S$ is full-dimensional.
So there is a matrix $C \in \R^{r \times d}$ such that $W = \augA \cdot C$.
For all $i \in [d]$ define $\chat_i \in \R^{r-1}$ so that $C^i = \begin{pmatrix} \chat_i \\ C_{r,i} \end{pmatrix}$.
Since $W$ is nonnegative, we have $A \chat_i + C_{r,i} b \ge 0$.
Observe that there is no $y \in \R^{r-1} \setminus \{0\}$ with $A y \ge 0$:
indeed, if $A y \ge 0$ for some nonzero~$y$ then for any $x \in P$ and any $t > 0$ we have $A (x + t y) + b \ge A x + b \ge 0$, implying that $P$ is unbounded, which is false.
We use this observation to show that $C_{r,i} > 0$.
\begin{itemize}
\item
Towards a contradiction, suppose $C_{r,i} = 0$.
Then $A \chat_i \ge 0$.
Since $W$ has no zero-columns, we have $\chat_i \ne 0$, contradicting the observation above.
\item
Towards a contradiction, suppose $C_{r,i} < 0$.
By dividing the inequality $A \chat_i + C_{r,i} b \ge 0$ by $-C_{r,i}$ we obtain
$ A (-\chat_i/C_{r,i}) - b \ge 0$.
Let $x \in P \setminus \{\chat_i / C_{r,i}\}$.
Then $A x + b \ge 0$, and by adding the previous inequality, we obtain $A (x - \chat_i/C_{r,i}) \ge 0$.
Since $x \ne \chat_i / C_{r,i}$, this also contradicts the observation above.
\end{itemize}
Thus we have shown that $C_{r,i} > 0$ for all $i \in [d]$.
Define a diagonal matrix $D \in \R^{d \times d}$ such that $D_{i,i} = C_{r_i} > 0$, and define $H' = D \cdot H$.
Then we have:
\begin{align*}
\augA
\begin{pmatrix}
 s_1 & s_2 & \cdots & s_m \\
 1   & 1   & \cdots & 1
\end{pmatrix}
&=
M
= W \cdot H
= \augA \, C \cdot H
= \augA \, C D^{-1} \cdot D H
\\
&=
\augA
\begin{pmatrix}
 \chat_1 / C_{r,1} & \chat_2 / C_{r,2} & \cdots & \chat_d / C_{r,d} \\
 1                 & 1                 & \cdots & 1
\end{pmatrix}
H'.
\end{align*}
Since the columns of $\augA$ are linearly independent, it follows:
\[
\begin{pmatrix}
 s_1 & s_2 & \cdots & s_m \\
 1   & 1   & \cdots & 1
\end{pmatrix}
=
\begin{pmatrix}
 \chat_1 / C_{r,1} & \chat_2 / C_{r,2} & \cdots & \chat_d / C_{r,d} \\
 1                 & 1                 & \cdots & 1
\end{pmatrix}
H'.
\]
Considering the last row, we see that each column of~$H'$ sums up to~$1$.
Since $H'$ is nonnegative, $H'$ is column-stochastic.
For $i \in [d]$ define $q_i = \chat_i / C_{r,i}$.
Then for all $j \in [m]$ we have $s_j \in \conv\{q_1, \cdots, q_d\}$.
Hence, defining the polytope $Q = \conv\{q_1, \cdots, q_d\}$ we have $S \subseteq Q$.
For all $i \in [d]$ we have $A \chat_i + C_{r,i} b \ge 0$, hence $A q_i + b \ge 0$.
It follows $Q \subseteq P$.
Thus $Q$ is nested between $P$~and~$S$.
\item[(b)]
Consider $d$~points $q_1, \ldots, q_d \in \R^{r-1}$ with $S \subseteq \conv\{q_1, \ldots, q_d\} \subseteq P$.
Define a matrix $W \in \Q^{n \times d}$ by $W^i = A q_i + b$ for $i \in [d]$.
Matrix~$W$ is nonnegative, as $q_i \in P$.
Define a column-stochastic matrix $H \in \Q^{d \times m}$ so that $s_j = \sum_{i=1}^d H_{i,j} q_i$ for $j \in [m]$.
Such~$H_{i,j}$ exist, as $s_j \in \conv\{q_1, \ldots, q_d\}$.
We have for all $j \in [m]$:
\begin{align*}
M^j 
& = A s_j + b && \text{definition of~$M$} \\
& = A \left( \sum_{i=1}^d H_{i,j} q_i \right) + b && \text{definition of~$H$} \\
& = \left( \sum_{i=1}^d A q_i \cdot H_{i,j} \right) + b \\
& = \sum_{i=1}^d (\underbrace{A q_i + b}_{W^i} ) \cdot H_{i,j} && \text{as $H$ is column-stochastic} \\
& = \left( W \cdot H \right)^j.
\end{align*}
Hence $M = W \cdot H$ is an NMF.
Since $S$ is full-dimensional, we have $\col(M) = \col((A \ \ b))$.
As $\col(W) \subseteq \col((A \ \ b))$ it follows that $\col(W) \subseteq \col(M)$, hence NMF $M = W \cdot H$ is restricted.
\end{enumerate} 

\section{Proofs of Section~\ref{sec-coverability}} \label{app-coverability}

\subsection{Discussion of Erroneous Claims in~\cite{DBLP:journals/tc/Bancilhon74}} \label{sub-Hippie-mistakes}

As mentioned in the main text, Bancilhon~\cite[Theorem~13]{DBLP:journals/tc/Bancilhon74} claims a statement that implies a positive answer to Paz's Question~\ref{question-Paz}.
We have shown that the correct answer to Paz's question is negative.
In the following we track down where the paper~\cite{DBLP:journals/tc/Bancilhon74} goes wrong.
The proof of~\cite[Theorem~13]{DBLP:journals/tc/Bancilhon74} offered therein relies on another (wrong) claim about \emph{cones}.

\newcommand{\C}{\mathcal{C}}%
\newcommand{\V}{\mathcal{V}}%
Let $\V$ be a vector space.
Let $v_1, \ldots, v_n \in \V$.
The \emph{polyhedral cone generated by $v_1, \ldots, v_n$} is the set
\[
\left\{
\sum_{i=1}^n \lambda_i v_i \;\middle\vert\; \lambda_1, \ldots, \lambda_n \ge 0
\right\} \subseteq \V\,.
\]
\begin{claim}[{\cite[Theorem~2]{DBLP:journals/tc/Bancilhon74}}, slightly paraphrased]
\label{claim-Hippie-Thm2}
 Let $\V$ be a vector space.
 Let $\V' \subseteq \V$ be a subspace of~$\V$.
 Let $\C \subseteq \V$ be a polyhedral cone generated by $n$~vectors that also span~$\V$.
 Then $\C \cap \V'$ is a polyhedral cone generated by at most $n$~vectors.
\end{claim}
Claim~\ref{claim-Hippie-Thm2} is false.
For a counterexample, consider the polyhedral cone~$\C$ generated by the following $6$~vectors:
\[
\begin{pmatrix}
1 \\ 0 \\ 0 \\ 1 \\ 0 \\ 0
\end{pmatrix}, \quad
\begin{pmatrix}
0 \\ 0 \\ 0 \\ 1 \\ 0 \\ 0
\end{pmatrix}, \quad
\begin{pmatrix}
0 \\ 1 \\ 0 \\ 0 \\ 1 \\ 0
\end{pmatrix}, \quad
\begin{pmatrix}
0 \\ 0 \\ 0 \\ 0 \\ 1 \\ 0
\end{pmatrix}, \quad
\begin{pmatrix}
0 \\ 0 \\ 1 \\ 0 \\ 0 \\ 1
\end{pmatrix}, \quad
\begin{pmatrix}
0 \\ 0 \\ 0 \\ 0 \\ 0 \\ 1
\end{pmatrix}
\]
It can be described equivalently by the conjunction of inequalities $0 \le x_1 \le x_4$ and $0 \le x_2 \le x_5$ and $0 \le x_3 \le x_6$.
Let $\V' \subset \R^6$ be the vector space defined by the equalities $x_4 = x_5 = x_6$.
Then the cone $\C' = \C \cap \V'$ can be described by the inequalities $0 \le x_1, x_2, x_3 \le x_4 = x_5 = x_6$.
The cone $\C'$ is generated by the following 8~vectors:
\[
\begin{pmatrix}
0 \\ 0 \\ 0 \\ 1 \\ 1 \\ 1
\end{pmatrix}, \quad
\begin{pmatrix}
0 \\ 0 \\ 1 \\ 1 \\ 1 \\ 1
\end{pmatrix}, \quad
\begin{pmatrix}
0 \\ 1 \\ 0 \\ 1 \\ 1 \\ 1
\end{pmatrix}, \quad
\begin{pmatrix}
0 \\ 1 \\ 1 \\ 1 \\ 1 \\ 1
\end{pmatrix}, \quad
\begin{pmatrix}
1 \\ 0 \\ 0 \\ 1 \\ 1 \\ 1
\end{pmatrix}, \quad
\begin{pmatrix}
1 \\ 0 \\ 1 \\ 1 \\ 1 \\ 1
\end{pmatrix}, \quad
\begin{pmatrix}
1 \\ 1 \\ 0 \\ 1 \\ 1 \\ 1
\end{pmatrix}, \quad
\begin{pmatrix}
1 \\ 1 \\ 1 \\ 1 \\ 1 \\ 1
\end{pmatrix}. \quad
\]
All those vectors are extremal in cone~$\C'$, so it cannot be generated by fewer than~$8$ vectors.
Hence, Claim~\ref{claim-Hippie-Thm2} is false.
%

Let us further examine how Claim~\ref{claim-Hippie-Thm2} was justified in~\cite{DBLP:journals/tc/Bancilhon74}.
The proof offered therein starts with the following claim, which is stated there without further justification:
\begin{claim}[{proof of \cite[Theorem~2]{DBLP:journals/tc/Bancilhon74}, slightly paraphrased}] \label{claim-Hippie-first-sentence}
A cone~$\C$ is generated by $n$~vectors if and only if it is limited by $n$~hyperplanes.
\end{claim}
Claim~\ref{claim-Hippie-first-sentence} is also false.
For a counterexample, consider the cone~$\C \subseteq \R_+^4$ limited by the following 6~hyperplanes:
\[
0 \le x_1 \le y, \quad 
0 \le x_2 \le y, \quad 
0 \le x_3 \le y.
\]
(All vectors in~$\C$ with $y=y^*$ form a cube of length~$y^*$.)
The following set contains 8~vectors, all of which are extremal in~$\C$:
\[
\{0,1\} \times \{0,1\} \times \{0,1\} \times \{1\} \quad\subset\quad \C.
\]
Hence $\C$ is not generated by 6~vectors, contradicting Claim~\ref{claim-Hippie-first-sentence}.

\subsection{Details of the Proof of Proposition~\ref{prop-LMC-determine}}\label{app-coverability-mainproposition}
In this subsection, we complete some details from the proof of Proposition~\ref{prop-LMC-determine}.

\begin{qproposition} {\ref{prop-LMC-determine}}
\stmtpropLMCdetermine
\end{qproposition}
\vspace{1em}

Here we will write $\vzero$ for the column vector with all zeros, whose dimension will be clear from the context. For every $p \in \N$, we denote by $\checkmark^p$ the $p$-fold concatenation of $\checkmark$ by itself.

Let us first take a detailed look at $\back \M$. For every $i \in [m]$ and $j \in [n]$, we have 
\[
(\back \M)_{i, b_j} = \mu(b_j)_{i} \cdot \mathbf{1} = \mu(b_j)_{i, m+1} = M_{j, i}.
\] 
\noindent From here it is easy to see that for every $w \in \Sigma^*$:
\[ (\back \M)^{w}= \begin{cases} 
      \vone & w=\varepsilon\\
      \frac{1}{m} \cdot e_{0} & w=a_{i} \text{ with } i \in [m] \\
      e_{m+1} & w=\checkmark^p \text{ with } p \in \mathbb{N}\\
      (0, M_{j}, 0)^{\top} & w=b_{j}\checkmark^p \text{ with } j \in [n], p \in \mathbb{N}_{0} \\
      \left(\frac{1}{m}M_{j, i}\right) \cdot e_{0} & w=a_{i}b_{j}\checkmark^p \text{ with } i \in [m], j \in [n], p \in \mathbb{N}_{0}\\
      \vzero & \text{otherwise}.
   \end{cases}
\]
 
\noindent That is,
\renewcommand{\kbldelim}{(}
\renewcommand{\kbrdelim}{)}
\[
\back \M
= \kbordermatrix{
    & \varepsilon & b_{1} & \cdots & b_{n} & \checkmark & a_{i}  & a_{i}b_{j} & b_{1}\checkmark  & \cdots & b_{n}\checkmark & \checkmark^2 & \cdots\\
 & 1 & 0 & \cdots & 0 & 0 & \frac{1}{m} & \frac{1}{m} M_{j, i} & 0 & \cdots & 0 & 0 & \cdots\\
 & 1 & M_{1, 1} & \cdots & M_{n, 1} & 0 & 0  & 0 & M_{1, 1} & \cdots & M_{n, 1} & 0 & \cdots\\
 & \vdots & \vdots & \ddots & \vdots & \vdots & \vdots  &  \vdots  & \vdots & \ddots & \vdots & \vdots & \cdots\\
 & 1 & M_{1, m} & \cdots & M_{n, m} & 0 & 0  & 0  & M_{1, m} & \cdots & M_{n, m} & 0 & \cdots\\
 & 1 & 0 & \cdots & 0  & 1 & 0 & 0 & 0 & \cdots & 0  & 1 & \cdots
  }
\]
Clearly, the following submatrix of $\back \M$ has the same column space as $\back \M$:
\[
 \kbordermatrix{
    & \varepsilon & b_{1} & \cdots & b_{n} & \checkmark\\
0        & 1 & 0 & \cdots & 0 & 0\\
1        & 1 & M_{1, 1} & \cdots & M_{n, 1} & 0\\
\vdots & \vdots & \vdots & \ddots & \vdots & \vdots\\
m        & 1 & M_{1, m} & \cdots & M_{n, m} & 0\\
m+1      & 1 & 0 & \cdots & 0  & 1}
\]
This implies that $\rk(\M) = \rk(M)+2 = r+2$.

For direction~(a), consider an NMF $M = W \cdot H$ and the LMC~$\M'$ defined in the main text.
For its backward matrix $\back \M^{\prime}\in \mathbb{R}_{+}^{\{0, 1, \ldots, d, d+1\} \times \Sigma^*}$, it is easy to see from the definition that 
for every $w \in \Sigma^*$:
\[ (\back \M^{\prime})^{w}= \begin{cases} 
      \mathbf{1} & w=\varepsilon\\
      \frac{1}{m} \cdot e_{0} & w=a_{i} \text{ with } i \in [m] \\
      e_{d+1} & w=\checkmark^p \text{ with } p \in \mathbb{N}\\
      (0, W_{j}, 0)^{\top} & w=b_{j}\checkmark^p \text{ with } j \in [n], p \in \mathbb{N}_{0} \\
      \left(\frac{1}{m} M_{j, i}\right) \cdot e_{0} & w=a_{i}b_{j}\checkmark^p \text{ with } i \in [m], j \in [n], p \in \mathbb{N}_{0}\\
      \vzero & \text{otherwise}.
   \end{cases}
\]
Indeed, for every $i \in [m]$:
\[ 
(\back \M^{\prime})^{a_{i}} = \left(\sum_{l \in [d]} \frac{1}{m} H_{l, i}\right) \cdot e_{0} = \frac{1}{m} \cdot e_{0},
\] 
since $H$ is a column-stochastic matrix. From here it is clear that the columns of the following submatrix of $\back \M^{\prime}$ span $\col (\back \M^{\prime})$:
\[
\kbordermatrix{
    & \varepsilon & b_{1} & \cdots & b_{n} & \checkmark \\
0 & 1 & 0 & \cdots & 0 & 0  \\
1 & 1 & W_{1, 1} & \cdots & W_{n, 1} & 0 \\
\vdots & \vdots & \vdots & \ddots & \vdots & \vdots\\
d & 1 & W_{1, d} & \cdots & W_{n, d} & 0 \\
d+1 & 1 & 0 & \cdots & 0  & 1 
  }
\]
Hence, $\rk(\M^{\prime}) = \rk(W) + 2$.
In particular, if NMF $M = W \cdot H$ is restricted, then $\rk(\M)  = \rk(M) + 2 = \rk(W) + 2 = \rk(\M^{\prime})$.

Now we show in more detail that $\M^{\prime} \ge \M$. We define a row-stochastic matrix:
\[
A = \begin{pmatrix}
1 & 0 & \cdots & 0 & 0\\
0 & H_{1, 1} & \cdots & H_{d, 1} & 0\\
\vdots & \vdots & \ddots & \vdots & \vdots\\
0 & H_{1, m} & \cdots & H_{d, m} & 0\\
0 & 0 & \cdots & 0 & 1
\end{pmatrix} \in \mathbb{R}_{+}^{\{0, 1, \ldots, m, m+1\} \times \{0, 1, \ldots, d, d+1\}}.
\]
For every $i \in [m]$ and $j \in [n]$:
\[
(A \cdot \back \M^{\prime})_{i, b_{j}} = A_{i} \cdot (\back \M^{\prime})^{b_{j}} = (H^{i})^{\top} \cdot (W_{j})^{\top} = (W_{j} \cdot H^{i})^{\top} = M_{j, i} = (\back \M)_{i, b_{j}}.
\]
\noindent From here it follows easily that  $A \cdot \back \M^{\prime} = \back \M$.
Hence, $\M^{\prime} \ge \M$.

Lastly we prove the other direction, (b). Let LMC $\M^{\prime} = (d+2, \Sigma, \mu^{\prime})$ be such that $\M^{\prime} \ge \M$. This implies that there exists a row-stochastic matrix $A$ such that $A \cdot \back \M^{\prime} = \back \M$. Given nonempty subsets $I$ and $J$ of the row and column indices of a matrix $C$, respectively, we write $C_{I, J}$ for the submatrix $(C_{i, j})_{i \in I, j \in J}$. It is clear from the description of $\back \M$ that $M$ has the following $(d+2)$-dimensional nonnegative factorization:
\begin{align}
M^{\top} = (\back \M)_{[m], \{b_{1}, \ldots, b_{n}\}} = A_{[m], \{0, 1, \ldots, d, d+1\}} \cdot (\back \M^{\prime})_{\{0, 1, \ldots, d, d+1\}, \{b_{1}, \ldots, b_{n}\}}. \label{eq:Mk+2dimNMF}
\end{align}
Assuming $M$ is a non-zero matrix, there exist  $i \in [m]$ and  $j \in [n]$ such that $M_{j, i} > 0$. Then,
\[
\frac{1}{m} M_{j, i} = (\back \M)_{0, a_{i} b_{j}} = A_{0} \cdot (\back \M^{\prime})^{a_{i} b_{j}} = \sum\limits_{i_{1}=0}^{d+1} A_{0, i_{1}} \cdot (\back \M^{\prime})_{i_{1}, a_{i} b_{j}} >0.
\]
Since $A$ and $(\back \M^{\prime})$ are nonnegative matrices, there exists $i_{1} \in \{0, 1, ..., d, d+1\}$ such that 
$A_{0, i_{1}} \cdot (\back \M^{\prime})_{i_{1}, a_{i} b_{j}} > 0$.
Without loss of generality, we may assume that $i_1 = 0$. That is,
\begin{align}
A_{0, 0} \cdot (\back \M^{\prime})_{0, a_{i} b_{j}} > 0. \label{ineq:zerotoab}
\end{align} 
Moreover, $(\back \M)_{m+1, \checkmark} = A_{m+1} \cdot (\back \M^{\prime})^{\checkmark} = 1$. The nonnegativity implies that there  exists $i_{2} \in \{0, 1, ..., d, d+1\}$ such that 
$A_{m+1, i_{2}} \cdot (\back \M^{\prime})_{i_{2}, \checkmark} > 0$. Note that $i_2 \neq 0$ since otherwise we would have that $(\back \M)_{m+1, a_{i} b_{j}} \ge A_{m+1, 0} \cdot (\back \M^{\prime})_{0, a_{i} b_{j}} >0$ which is a contradiction. We may therefore, without loss of generality, assume that $i_2 = d+1$, i.e., 
\begin{align}
A_{m+1, d+1} \cdot (\back \M^{\prime})_{d+1, \checkmark} > 0. \label{ineq:lastcheckmark}
\end{align}

\begin{lemma}\label{lemma:reductionNMFtoPAmin:eliminate2cols}
It holds that $M^{\top} = A_{[m], [d]} \cdot (\back \M^{\prime})_{[d], \{b_{1}, \ldots, b_{n}\}}$,
where matrix $A_{[m], [d]}$ is row-stochastic.
\end{lemma}
\begin{proof}
By (\ref{eq:Mk+2dimNMF}), it suffices to show that \[
A_{[m], \{0, 1, \ldots, d, d+1\}} \cdot (\back \M^{\prime})_{\{0, 1, \ldots, d, d+1\}, \{b_{1}, \ldots, b_{n}\}}= A_{[m], [d]} \cdot (\back \M^{\prime})_{[d], \{b_{1}, \ldots, b_{n}\}}.
\]
That is, for any $l_{1} \in [m]$ and  $l_{2} \in [n]$ we need to show that
\[
A_{l_{1}, \{0, 1, \ldots, d, d+1\}} \cdot (\back \M^{\prime})_{\{0, 1, \ldots, d, d+1\}, b_{l_{2}}}= A_{l_{1}, [d]} \cdot (\back \M^{\prime})_{[d], b_{l_{2}}}.
\]
To do this, it suffices to show that $A_{l_{1}, 0} \cdot (\back \M^{\prime})_{0, b_{l_{2}}} + A_{l_{1}, d+1} \cdot (\back \M^{\prime})_{d+1, b_{l_{2}}} = 0$. In the following, we prove that $A_{l_{1}, 0} = 0$ and $A_{l_{1}, d+1} = 0$. By an analogous argument, it also holds that $(\back \M^{\prime})_{0, b_{l_{2}}} = 0$ and $(\back \M^{\prime})_{d+1, b_{l_{2}}} = 0$.

If $A_{l_{1}, 0} >0$, then from (\ref{ineq:zerotoab}) it would follow that
\begin{align*}
(\back \M)_{l_{1}, a_{i} b_{j}} = A_{l_{1}} \cdot (\back \M^{\prime})^{a_{i} b_{j}} \ge A_{l_{1}, 0} \cdot (\back \M^{\prime})_{0, a_{i} b_{j}} > 0,
\end{align*} 
which is a contradiction since $(\back \M)_{l_{1}, a_{i} b_{j}} = 0$ for all $l_{1} \in [m]$. 
Similarly, if $A_{l_{1}, d+1} > 0$ then by (\ref{ineq:lastcheckmark}) we would have that 
\begin{align*}
(\back \M)_{l_{1}, \checkmark} = A_{l_{1}} \cdot (\back \M^{\prime})^{\checkmark} \ge A_{l_{1}, d+1} \cdot (\back \M^{\prime})_{d+1, \checkmark} > 0,
\end{align*} 
which is a contradiction since $(\back \M)_{l_{1}, \checkmark} = 0$ for all $l_{1} \in [m]$. Hence, $A_{l_{1}, 0} = A_{l_{1}, d+1} =0$. 

We have thus shown that  $A_{[m], \{0\}} = A_{[m], \{d+1\}} =  (0, \ldots, 0)^{\top}$. Since $A$ is row stochastic, this implies that $A_{[m], [d]}$ is row-stochastic.
\end{proof}

From Lemma~\ref{lemma:reductionNMFtoPAmin:eliminate2cols} we get a  $d$-dimensional nonnegative factorizaton $M = W \cdot H$ where $W = ((\back \M^{\prime})_{[d], \{b_{1}, \ldots, b_{n}\}})^{\top}$ and $H = (A_{[m], [d]})^{\top}$. 

Suppose $\M^{\prime}$ and $\M$ have the same rank. Then $\row(\back \M) = \row(\back \M^{\prime})$, as we already knew that $\row(\back \M) \subseteq \row(\back \M^{\prime})$ since $A \cdot \back \M^{\prime} = \back \M$. Take any $v \in \col(W)$. Since $W^{\top} = (\back \M^{\prime})_{[d], \{b_{1}, \ldots, b_{n}\}}$, there exists $u \in \row(\back \M^{\prime})$ such that $v^{\top} = u_{\{b_{1}, \ldots, b_{n}\}}$. As $\row(\back \M^{\prime}) = \row(\back \M)$, from $\back \M$ it is clear that $v \in \col(M)$. Hence, $\col(W) \subseteq \col(M)$ and therefore $\rk(W) \le \rk(M)$. Since $M = W \cdot H$, this implies that $\rk(W) = \rk(M)$ as required. This completes the proof of Proposition~\ref{prop-LMC-determine}.

\section{Proofs of Section~\ref{sec-NPP32}} \label{app-NPP2d}

First we prove Lemma~\ref{lem:rationallinfract} from the main text.

\begin{qlemma}{\ref{lem:rationallinfract}}
\stmtlemrationallinfract
\end{qlemma}
\begin{proof}
Let vertices $u$ and $v$ lie on edges $p_{1} p_{1}^{\prime}$ and~$p_{2} p_{2}^{\prime}$ of $P$, respectively. Let $s_{1}$ be the point where the supporting line segment $u v$ touches the inner polygon $S$ on its left. Let $\lambda_1$  and $\lambda_2$ be the convex representations of $u$ and $v$, respectively. That is, $u = (1-\lambda_1) p_1+ \lambda_1 p_{1}^{\prime}$ and $v = (1-\lambda_2) p_2 + \lambda_2 p_{2}^{\prime}$. By definition of the ray function $r$ we have $\lambda_2 = r(\lambda_1)$.
 
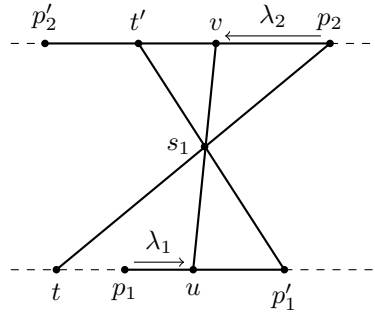
\begin{figure}[t]
\begin{center}
    \centering
\begin{tikzpicture}[xscale=.6]

\draw[dashed] (0,0)  -- (2.5,0);
\draw[thick] (2.5,0) -- (6,0);
\draw[dashed] (6,0) -- (8,0);
\draw[dashed] (0,3)  -- (.75,3);
\draw[thick] (.75,3) -- (7,3);
\draw[dashed] (7,3) -- (8,3);
\draw[thick] (1,0) -- (7,3);
\draw[thick] (4,0) -- (4.5,3);
\draw[thick] (6,0) -- (2.8,3);

\draw (1,0) node[circle,fill,inner sep=1pt,label=below:$t$](B){};
\draw (2.5,0) node[circle,fill,inner sep=1pt,label=below:$p_1$](v1){};
\draw (4,0) node[circle,fill,inner sep=1pt,label=below:$u$](C){};
\draw (6,0) node[circle,fill,inner sep=1pt,label=below:$p_{1}^{\prime}$](v2){};

\draw (7,3) node[circle,fill,inner sep=1pt,label=above:$p_2$](v3){};
\draw (.75,3) node[circle,fill,inner sep=1pt,label=above:$p_{2}^{\prime}$](v4){};
\draw (4.5,3) node[circle,fill,inner sep=1pt,label=above:$v$](CC){};
\draw (2.8,3) node[circle,fill,inner sep=1pt,label=above:$t^{\prime}$](D){};
\draw (4.25,1.63125) node[circle,fill,inner sep=1pt,label=left:$s_1$](M1){};

\draw[->] (2.7,0.1) --  (3.8,0.1) node[midway,above] {$\lambda_1$};
\draw[->] (6.8,3.1) -- (4.7,3.1) node[midway,above] {$\lambda_2$};

\end{tikzpicture}
\end{center}
 \caption{Lines~$p_{1} p_{1}^{\prime}$ and~$p_{2} p_{2}^{\prime}$ are parallel. 
\label{fig:parallelLines}}
\end{figure}

Let us first consider the case when lines $p_{1} p_{1}^{\prime}$ and~$p_{2} p_{2}^{\prime}$ are parallel; see Figure~\ref{fig:parallelLines} for illustration. Let $t$ denote the intersection of lines $p_{1} p_{1}^{\prime}$ and $s_1 p_2$, and let $t^{\prime}$ denote the intersection of lines $p_{2} p_{2}^{\prime}$ and $s_1 p_{1}^{\prime}$.  Since $s_1 \in \mathbb{Q}^{2}$, we have $t = (1- b) \cdot p_1 + b \cdot p_{1}^{\prime}$ and $t^{\prime} = (1- d) \cdot p_2 + d \cdot p_{2}^{\prime}$ for some $b, d \in \mathbb{Q}$. Numbers $b, d$ can be computed from the input vertices $p_{1}, p_{1}^{\prime}, p_{2}, p_{2}^{\prime}$ using a constant number of arithmetic operations and therefore have bit-length $O(L)$. Since $\triangle t u s_1 \sim \triangle p_2 v s_1$ and $\triangle t p_{1}^{\prime} s_1 \sim \triangle  p_2 t^{\prime} s_1$, we have $\frac{u t}{v p_2} = \frac{t s_{1}}{s_{1} p_2}$ and $\frac{t s_{1}}{s_{1} p_2} = \frac{p_{1}^{\prime} t}{t^{\prime} p_2}$. Hence, $\frac{u t}{v p_2} = \frac{p_{1}^{\prime} t}{t^{\prime} p_2}$ and therefore $\frac{\lambda_1 - b}{\lambda_2} = \frac{1 - b}{d}$. Hence, 
\begin{align*}
\lambda_2 = \frac{d}{1 - b} \cdot \lambda_1 - \frac{b d}{1 - b}
\end{align*} 
where the coefficients $\frac{d}{1 - b}, \frac{b d}{1 - b} \in \mathbb{Q}$ have bit-length $O(L)$.

\begin{figure}[t]
\begin{center}
    \centering
\begin{tikzpicture}[xscale=.6]

\draw[thick] (0,0)  -- (5,0);
\draw[dashed] (5,0) -- (10,0);
\draw[dashed] (9,0) -- (4,3.35);
\draw[thick] (4,3.35) -- (1,5.25);
\draw[thick] (2,4.615) -- (1,0);
\draw[dashed] (6,2) -- (1.4,2);

\draw (0,0) node[circle,fill,inner sep=1pt,label=below:$p_1$](v1){};
\draw (1,0) node[circle,fill,inner sep=1pt,label=below:$u$](C){};
\draw (5,0) node[circle,fill,inner sep=1pt,label=below:$p_{1}^{\prime}$](v2){};
\draw (9,0) node[circle,fill,inner sep=1pt,label=below:$t$](B){};
\draw (6,2) node[circle,fill,inner sep=1pt,label=below:$t^{\prime}$](D){};
\draw (4,3.35) node[circle,fill,inner sep=1pt,label=left:$p_2$](v_3){};
\draw (2,4.615) node[circle,fill,inner sep=1pt,label=left:$v$](CC){};
\draw (1,5.25) node[circle,fill,inner sep=1pt,label=left:$p_{2}^{\prime}$](v_4){};
\draw (1.4,2) node[circle,fill,inner sep=1pt,label=left:$s_1$](M1){};

\draw[->] (0.1,0.1) --  (.9,0.1) node[midway,above] {$\lambda_1$};
\draw[->] (9.1,0.1) -- (6.1,2.1) node[midway,above] {$d$};
\draw[->] (5.7,2.1) -- (1.6,2.1) node[midway,above] {$c$};
\draw[->] (3.95,3.55) -- (2.15,4.7) node[midway,right] {$\lambda_2$};

\end{tikzpicture}
\end{center}
 \caption{Lines~$p_{1} p_{1}^{\prime}$ and~$p_{2} p_{2}^{\prime}$  intersect at point~$t$. \label{fig:intersectingLines}}
\end{figure}
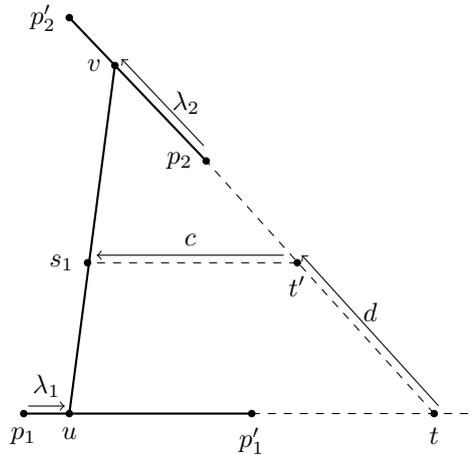

Let us now consider the second case: when lines $p_{1} p_{1}^{\prime}$ and~$p_{2} p_{2}^{\prime}$  are not parallel. Let $t$ denote their intersection; see Figure \ref{fig:intersectingLines} for illustration. Note that $t \in \mathbb{Q}^2$. Without loss of generality, let $p_{1}^{\prime}$ be a convex combination of $\{p_{1}, t\}$ and $p_{2}$ be a convex combination of $\{t, p_{2}^{\prime}\}$. Then $p_{1}^{\prime} =  (1- a) \cdot p_{1} + a \cdot t$ and $p_{2} =  (1- b) \cdot t + b \cdot p_{2}^{\prime}$ for some $a, b \in [0, 1]$. Numbers $a, b$ are rational as they are each the unique solution of a linear system with rational coefficients. Moreover $a, b$ can be computed from the input vertices $p_{1}, p_{1}^{\prime}, p_{2}, p_{2}^{\prime}$ using a constant number of arithmetic operations and therefore have bit-length $O(L)$. Using the above representation of $p_{1}^{\prime}$ as a convex combination of $\{p_1,  t\}$, we get
\begin{align*}
u = (1-\lambda_1) \cdot p_1 + \lambda_1 \cdot ((1- a) \cdot p_1 + a \cdot t) = (1- a \lambda_1) \cdot p_1 + a \lambda_1 \cdot t
\end{align*}
and therefore $u - t = (1- a \lambda_1) \cdot (p_1 - t)$. Similarly, we have
\begin{align*}
v = (1 - b - \lambda_2 + \lambda_2 b) \cdot t + (b + \lambda_2 - \lambda_2 b) \cdot p_{2}^{\prime}
\end{align*}
and therefore $v - t = (b + \lambda_2 - \lambda_2 b) \cdot (p_{2}^{\prime} - t)$.
Let the line through~$s_1$ parallel with $p_{1} t$ intersect $p_{2}^{\prime} t$ at point $t^{\prime}$. Note that $s_1 - t^{\prime} = c \cdot (p_{1} - t)$ and $t^{\prime} - t = d \cdot (p_{2}^{\prime} - t)$ for some $c, d \in [0, 1] \cap \mathbb{Q}$. Numbers $c, d$ can be computed from vertices $p_{1}, p_{1}^{\prime}, p_{2}, p_{2}^{\prime}$ using a constant number of arithmetic operations and thus have bit-length $O(L)$. 
Since $\triangle v t^{\prime} s_1 \sim \triangle  v t u$, it holds that $\frac{v t^{\prime}}{v t} = \frac{s_{1} t^{\prime}}{u t}$. We therefore have that
\begin{align*}
\frac{b + \lambda_2 - \lambda_2 b - d}{b + \lambda_2 - \lambda_2 b} = \frac{c}{1- a \lambda_1}.
\end{align*}
From the above equation we get that
\begin{align*}
\lambda_2 = \frac{a (b - d) \cdot \lambda_1 + b (c - 1) + d}{a (b - 1) \cdot \lambda_1 + (b -1) (c - 1)}
\end{align*}
where $a (b - d), b (c - 1) + d, a (b - 1), (b -1) (c - 1) \in \mathbb{Q}$ have bit-length $O(L)$.
\end{proof}

A function $f : I \to \mathbb{R}$, where $I \subseteq \mathbb{R}$, is called a \emph{rational linear fractional transformation} if there exist $a, b, c, d \in \mathbb{Q}$ such that $f(\lambda) = \frac{a\lambda + b}{c\lambda+d}$ for every $\lambda \in I$.

\begin{lemma} \label{rationallft_closed_composition}
Given two rational linear fractional transformations $f_{1}(\lambda) = \frac{a_{1} \lambda + b_{1}}{c_{1}\lambda+d_{1}}$ and $f_{2}(\lambda) = \frac{a_{2} \lambda + b_{2}}{c_{2} \lambda+d_{2}}$, their composition $f_{2} \circ f_{1}$ is also a rational linear fractional transformation. Specifically, $(f_{2} \circ f_{1}) (\lambda) = \frac{a\lambda + b}{c\lambda+d}$ where $a, b, c, d \in \mathbb{Q}$ are such that
\[
\begin{bmatrix}
   a ~& b \\
   c ~& d 
\end{bmatrix}
=
\begin{bmatrix}
    a_{2} & b_{2} \\
    c_{2} & d_{2} 
\end{bmatrix}
\cdot 
\begin{bmatrix}
    a_{1} & b_{1} \\
    c_{1} & d_{1} 
\end{bmatrix}.
\]
\end{lemma}
\begin{proof}
For every $\lambda$ we have
\begin{align*}
(f_{2} \circ f_{1}) (\lambda) &= \frac{a_{2} \frac{a_{1} \lambda + b_{1}}{c_{1}\lambda+d_{1}} + b_{2}}{c_{2} \frac{a_{1} \lambda + b_{1}}{c_{1}\lambda+d_{1}} + d_{2}}\\ & = \frac{a_{2} (a_{1} \lambda + b_{1}) + b_{2}(c_{1}\lambda+d_{1})}{c_{2} (a_{1} \lambda + b_{1}) + d_{2}(c_{1}\lambda+d_{1})} = \frac{(a_{2} a_{1} + b_{2} c_{1})\lambda + (a_{2} b_{1} + b_{2} d_{1})}{(c_{2} a_{1} + d_{2} c_{1} )\lambda + (c_{2}b_{1} + d_{2}d_{1}) }.
\end{align*}
\end{proof}

\begin{corollary}\label{cor:bitsizeNestedVertices}
Given $r_{i}(\lambda)= \frac{a_i \lambda + b_i}{c_i \lambda + d_i}$ in terms of $a_i, b_i, c_i, d_i \in \mathbb{Q}$, for $i \in [k]$, one can compute in polynomial time $a, b, c, d \in \mathbb{Q}$ such that $r(\lambda) = (r_k \circ \ldots \circ r_2 \circ r_1)(\lambda) = \frac{a\lambda + b}{c\lambda+d}$. 
\end{corollary}
\begin{proof}
By Lemma~\ref{rationallft_closed_composition}, the coefficients of the composition $r_k \circ \ldots \circ r_2 \circ r_1$ are computable by iterative matrix multiplication:
 \[
 \begin{bmatrix}
    a ~& b \\
    c ~& d 
 \end{bmatrix}
 =
 \begin{bmatrix}
     a_{k} & b_{k} \\
     c_{k} & d_{k} 
 \end{bmatrix}
 \cdot \ldots \cdot
 \begin{bmatrix}
     a_{2} & b_{2} \\
     c_{2} & d_{2} 
 \end{bmatrix}
 \cdot 
 \begin{bmatrix}
     a_{1} & b_{1} \\
     c_{1} & d_{1} 
 \end{bmatrix}.
 \]
\noindent The result follows.
\end{proof}

\section{Additional Figures for Section~\ref{sec-NPP3d}} \label{app-NPP3d}

\begin{figure}
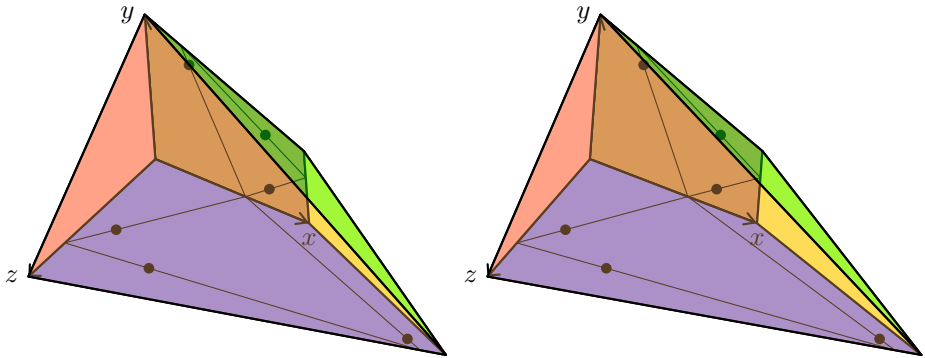

\begin{center}
\objectThreeD{32}{40}{5}{2.5}\hspace{-13mm}
\objectThreeD{32}{40}{-5}{2.5}
\end{center}
\caption{Two projections of~$P$ suitable for ``cross-eyed'' stereoscopic viewing.}
\label{fig-3d-cross-eyed}
\end{figure}

Figure~\ref{fig-3d-cross-eyed} shows a variant of Figure~\ref{fig-3d-parallel-eyed}, suitable for ``cross-eyed'' stereoscopic viewing.

\begin{figure}
\begin{center}
\begin{tikzpicture}[
scale=5,
dot/.style={circle,fill=black,minimum size=4pt,inner sep=0pt,outer sep=-1pt},
spy using outlines={rectangle, magnification=10, width=3.5cm, height=3cm, connect spies}
]
\draw[->,>=stealth'] (0,0) -- (2.4,0) node[right] {$x$};
\draw[->,>=stealth'] (0,0) -- (0,1.35) node[above] {$z$};
\foreach \x in {0.25,0.5,0.75,1.0,1.25,1.5,1.75,2.0,2.25}
    \draw[shift={(\x,0)}] (0pt,0pt) -- (0pt,-0.4pt) node[below] {$\x$};
\foreach \x in {0.25,0.5,0.75,1.0,1.25}
    \draw[shift={(0,\x)}] (0,0) -- (-0.4pt,0) node[left] {$\x$};
\draw[thick] (0,0) coordinate (O1) -- (1,0) coordinate (O2) -- (9/4,1/2) coordinate (O5) -- (0,8/7) coordinate (O6) -- (O1);
\coordinate (I4) at (2,1/2);
\coordinate (I5) at (1/2,3/4);
\coordinate (I6) at (1/6,7/12);
\node[dot,label={[label distance=0mm]-90:$s_4$}] at (I4) {};
\node[dot,label={[label distance=0mm]90:$s_5$}] at (I5) {};
\node[dot,label={[label distance=0mm]0:$s_6$}] at (I6) {};
\coordinate (M1) at ({2-sqrt(2)}, 0);
\coordinate (M4) at  (intersection cs:
       first line={(M1)--(I4)},
       second line={(O5)--(O6)});
\coordinate (M5) at  (intersection cs:
       first line={(M4)--(I5)},
       second line={(O1)--(O6)});
\draw (M1) -- (M4) -- (M5) -- (M1);
\draw (M1) -- ({2-sqrt(2)}, -0.12) node[below,xshift=27mm,yshift=1mm] {$q_1^* = (2-\sqrt{2},0,0)^\top \approx (0.5858,0,0)^\top$};
\node[label={[label distance=-1mm]90:$q_4^*$}] at (M4) {};
\node[label={[label distance=-1.5mm,yshift=1mm]180:$q_5^*$}] at (M5) {};

\coordinate (M1+) at (1, 0);
\coordinate (M4+) at  (intersection cs:
       first line={(M1+)--(I4)},
       second line={(O5)--(O6)});
\coordinate (M5+) at  (intersection cs:
       first line={(M4+)--(I5)},
       second line={(O1)--(O6)});
\coordinate (M6+) at  (intersection cs:
       first line={(M5+)--(I6)},
       second line={(O1)--(O2)});
\draw[dashed] 
(M6+) -- (M5+) -- (M4+) -- (M1+);

\coordinate (M1-) at (0., 0);
\coordinate (M4-) at  (intersection cs:
       first line={(M1-)--(I4)},
       second line={(O5)--(O6)});
\coordinate (M5-) at  (intersection cs:
       first line={(M4-)--(I5)},
       second line={(O1)--(O6)});
\coordinate (M6-) at  (intersection cs:
       first line={(M5-)--(I6)},
       second line={(O1)--(O2)});
\draw[dotted] (M1-) -- (M4-) -- (M5-) -- (M6-);

\coordinate (M1spy) at (0.580,0.030);
\spy [red] on (M1spy) in node [right] at (1.5,1.1);

\end{tikzpicture}
\end{center}
\caption{Detailed view on the $xz$-plane.
The image shows that $q_1^*$ cannot be moved right on the $x$-axis without increasing the number of vertices of the nested polytope.
}
\label{fig-convex-display}
\end{figure}
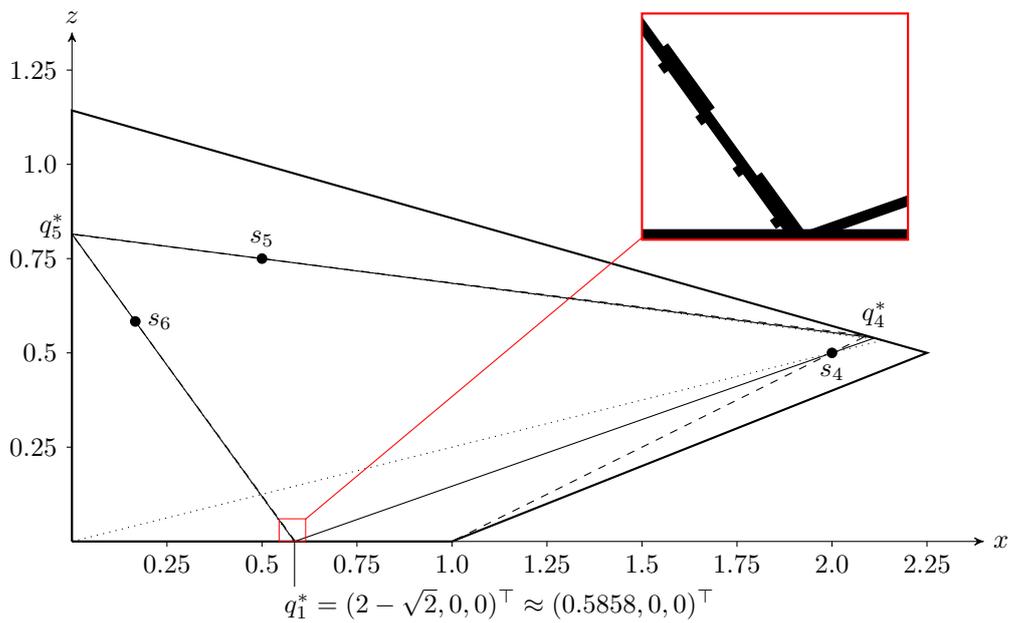

Figure~\ref{fig-convex-display} shows the $xz$-plane, in the same way as Figure~\ref{fig-concave-display} shows the $xy$-plane.

\begin{figure}
\begin{center}
\begin{tikzpicture}[xscale=11,yscale=4]
\draw[->, >=stealth'] (0,0) -- (1.1,0) node[right] {$\lambda$};
\foreach \x in {0.25,0.5,0.75,1.0}
    \draw[shift={(\x,0)}] (0pt,0pt) -- (0pt,-0.4pt) node[below] {$\x$};
\draw[->, >=stealth'] (0,-0.6) -- (0,0.7) node[above] {}; 
\foreach \x in {-0.5,-0.25,0,0.25,0.5}
    \draw[shift={(0pt,\x)}] (0pt,0pt) -- (-0.2pt,0) node[left] {$\x$};
\draw[smooth,samples=194,domain=0:1]
   plot(\x,{40*(4*\x-2-\x^2)/(40*\x-137)});
\node[right] at (0.25,0.4) {$s(\lambda)$};
\draw ({2-sqrt(2)},0) -- ({2-sqrt(2)}, -4pt) node[below] {$2-\sqrt{2} \approx 0.5858$};
\end{tikzpicture}

\end{center}
\caption{The slack function $s(\lambda) = \frac{23 \lambda - 80}{40 \lambda - 137} - \lambda$ corresponding to Figure~\ref{fig-convex-display}.
When $s(\lambda)<0$, there is no nested triangle with vertex $(\lambda,0,0)$.}
\label{fig-convex-slack}
\end{figure}
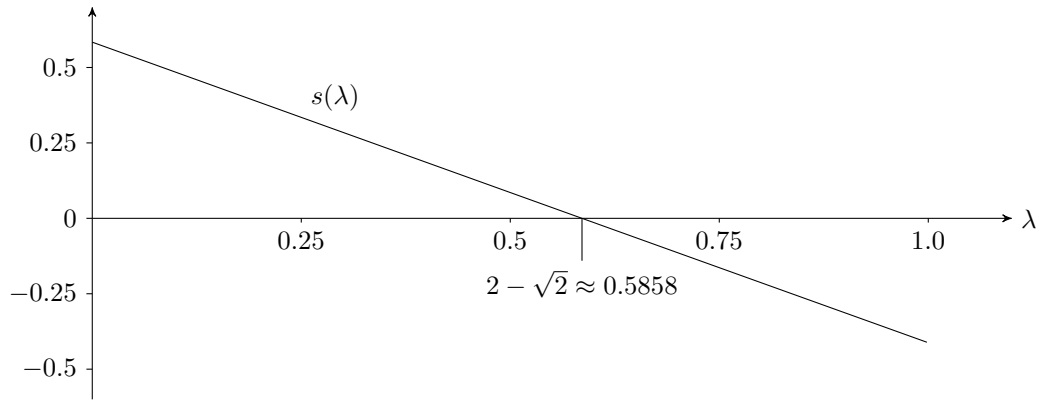 

Figure~\ref{fig-convex-slack} plots the corresponding slack function.
It shows that a sign change from positive to negative occurs at the root $\lambda = 2 - \sqrt{2}$.

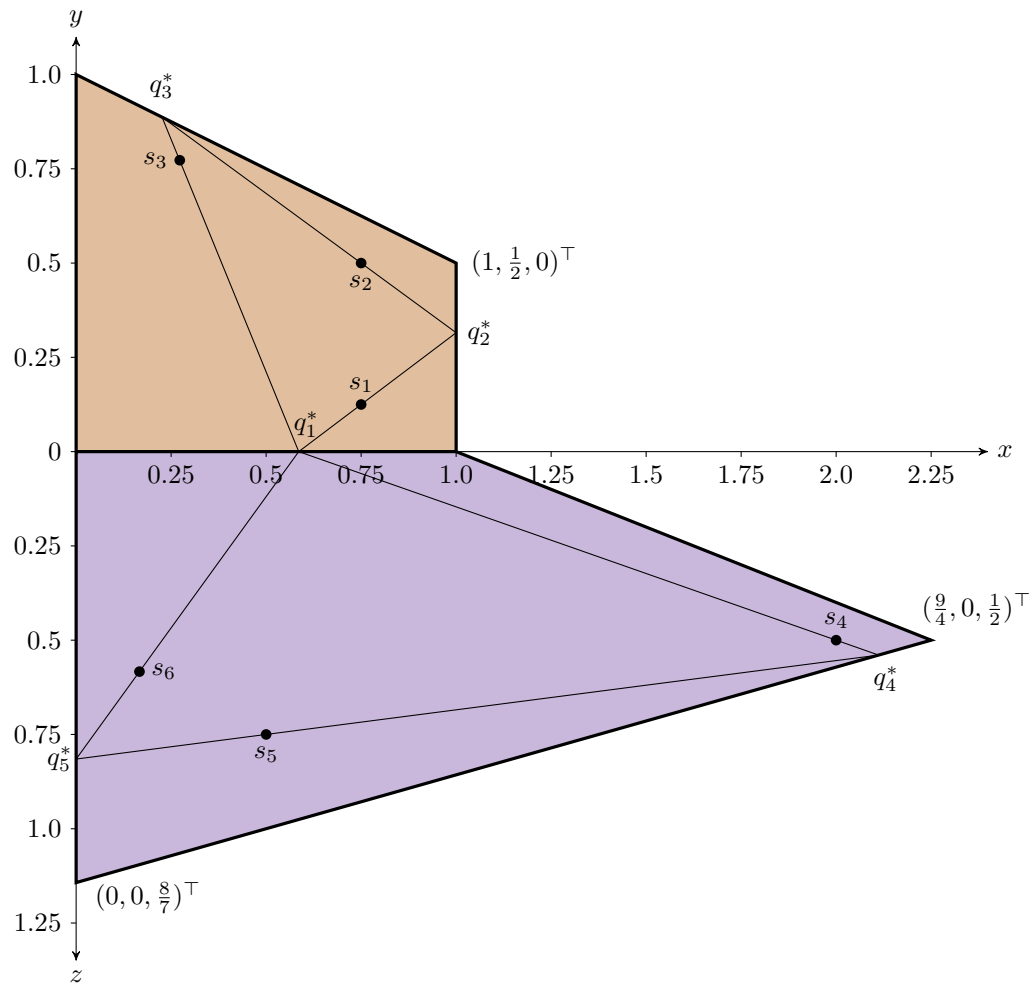
\begin{figure}
\begin{center}
\begin{tikzpicture}[
scale=5,
dot/.style={circle,fill=black,minimum size=4pt,inner sep=0pt,outer sep=-1pt},
]
    
\draw[very thick,fill=brown!95!orange!50] (0,0) coordinate (O1) -- (1,0) coordinate (O2) -- (1,1/2) coordinate (O3) -- (0,1) coordinate (O4) -- (O1);
\coordinate (I1) at (3/4,1/8);
\coordinate (I2) at (3/4,1/2);
\coordinate (I3) at (3/11,17/22);
\coordinate (M1) at ({2-sqrt(2)}, 0);
\coordinate (M2) at  (intersection cs:
       first line={(M1)--(I1)},
       second line={(O2)--(O3)});
\coordinate (M3) at  (intersection cs:
       first line={(M2)--(I2)},
       second line={(O3)--(O4)});
\draw (M1) -- (M2) -- (M3) -- (M1);

\draw[very thick,fill=blue!60!orange!35] (0,0) coordinate (O1) -- (1,0) coordinate (O2) -- (9/4,-1/2) coordinate (O5) -- (0,-8/7) coordinate (O6) -- (O1);
\coordinate (I4) at (2,-1/2);
\coordinate (I5) at (1/2,-3/4);
\coordinate (I6) at (1/6,-7/12);
\coordinate (M4) at  (intersection cs:
       first line={(M1)--(I4)},
       second line={(O5)--(O6)});
\coordinate (M5) at  (intersection cs:
       first line={(M4)--(I5)},
       second line={(O1)--(O6)});
\draw (M1) -- (M4) -- (M5) -- (M1);

\draw[->,>=stealth'] (0,0) -- (2.4,0) node[right] {$x$};
\draw[->, >=stealth'] (0,0) -- (0,1.10) node[above] {$y$};
\draw[->,>=stealth'] (0,0) -- (0,-1.35) node[below] {$z$};
\foreach \x in {0.25,0.5,0.75,1.0,1.25,1.5,1.75,2.0,2.25}
    \draw[shift={(\x,0)}] (0pt,0pt) -- (0pt,-0.4pt) node[below] {$\x$};
\foreach \x in {0,0.25,0.5,0.75,1.0}
    \draw[shift={(0,\x)}] (0,0) -- (-0.4pt,0) node[left] {$\x$};
\foreach \x in {0.25,0.5,0.75,1.0,1.25}
    \draw[shift={(0,-\x)}] (0,0) -- (-0.4pt,0) node[left] {$\x$};

\node[label={[label distance=-0.5mm,xshift=1mm]90:$q_1^*$}] at (M1) {};
\node[label={[label distance=-1mm]0:$q_2^*$}] at (M2) {};
\node[label={90:$q_3^*$}] at (M3) {};
\node[label={[label distance=-1mm,xshift=1mm]-90:$q_4^*$}] at (M4) {};
\node[label={[label distance=-1mm,xshift=1mm]180:$q_5^*$}] at (M5) {};
\node[dot,label={[label distance=0mm]90:$s_1$}] at (I1) {};
\node[dot,label={[label distance=0mm]-90:$s_2$}] at (I2) {};
\node[dot,label={[label distance=0mm]180:$s_3$}] at (I3) {};
\node[dot,label={[label distance=0mm]90:$s_4$}] at (I4) {};
\node[dot,label={[label distance=0mm]-90:$s_5$}] at (I5) {};
\node[dot,label={[label distance=0mm]0:$s_6$}] at (I6) {};
\node[label={[label distance=-0.5mm,xshift=0mm]0:$(1,\frac12,0)^\top$}] at (O3) {};
\node[label={[label distance=-0.5mm,xshift=6mm]90:$(\frac94,0,\frac12)^\top$}] at (O5) {};
\node[label={[xshift=0mm,yshift=-2mm]0:$(0,0,\frac87)^\top$}] at (O6) {};

%
%
%
%
\end{tikzpicture}
\end{center}
\caption{Combined view of the $xy$-plane and the $xz$-plane.
One may imagine that the figure is folded along the $x$-axis.
}
\label{fig-combined}
\end{figure} 

Figure~\ref{fig-combined} provides a combined view of the $xy$-plane and the $xz$-plane, with the coordinates of some vertices of~$P$.

}{}

\end{document}